\documentclass[a4paper,USenglish,numberwithinsect,nosubfigcap]{lipics-v2019}

\usepackage{cite}
\graphicspath{{./figs/}}

\title{Long Alternating Paths Exist}

\author{Wolfgang Mulzer}{Institut f\"ur Informatik, 
Freie Universit\"at Berlin, Takustra\ss{}e 9, 14195 Berlin, Germany}
{mulzer@inf.fu-berlin.de}{https://orcid.org/0000-0002-1948-5840}{Supported 
in part by ERC StG 757609.}

\author{Pavel Valtr}{Department of Applied Mathematics, Faculty of Mathematics and Physics, 
Charles University, Prague, Czech Republic}
{valtr@kam.mff.cuni.cz}{}{}

\authorrunning{W. Mulzer and P. Valtr}

\Copyright{Wolfgang Mulzer and Pavel Valtr}

\ccsdesc[100]{Theory of computation~Computational geometry}

\keywords{Non-crossing path, bichromatic point sets}

\category{}

\supplement{}

\funding{Research partly supported by the German
Research Foundation within the collaborative DACH project \emph{Arrangements and Drawings} 
as DFG Project MU-3501/3-1.
Work by P. Valtr was supported by the grant no.~18-19158S of the Czech Science Foundation (GA\v{C}R).}

\acknowledgements{This work was initiated 
at the second DACH workshop on Arrangements 
and Drawings which took place 21.--25.~January 2019
at Schloss St.~Martin, Graz, Austria.
We would like to thank the organizers and all the participants
of the workshop for creating a conducive 
research atmosphere and for stimulating discussions.
We also thank Zolt\'an Kir\'aly for pointing out the reference~\cite{MullnerR19} to us.}

\nolinenumbers 

\hideLIPIcs  

\EventEditors{Sergio Cabello and Danny Z. Chen}
\EventNoEds{2}
\EventLongTitle{36th International Symposium on Computational Geometry (SoCG 2020)}
\EventShortTitle{SoCG 2020}
\EventAcronym{SoCG}
\EventYear{2020}
\EventDate{June 23--26, 2020}
\EventLocation{Z\"{u}rich, Switzerland}
\EventLogo{socg-logo}
\SeriesVolume{164}
\ArticleNo{57}

\usepackage{mathtools}
\usepackage{xspace}
\usepackage{tikz}

\newcommand{\N}{\mathbb{N}}

\newcommand{\R}{\mathbb{R}}

\newcommand{\ignore}[1]{}

\newcommand{\eps}{\varepsilon}

\begin{document}
\maketitle

\begin{abstract}
Let $P$ be a set of $2n$ points in convex
position, such that $n$ points are colored red 
and $n$ points
are colored blue. A non-crossing alternating path 
on $P$ of length $\ell$ is a sequence $p_1, \dots,
p_\ell$ of $\ell$ points from $P$ so that 
(i) all points are pairwise distinct; 
(ii) any two consecutive points $p_i, p_{i+1}$ have
different colors; and 
(iii) any two segments $p_i p_{i+1}$ and 
$p_j p_{j+1}$ have disjoint relative interiors, 
for $i \neq j$.

We show that there is an absolute constant
$\varepsilon > 0$, independent of $n$ and of 
the coloring, such that $P$ always admits a 
non-crossing alternating path of length at 
least $(1 + \eps)n$. The result is obtained 
through a slightly stronger statement: 
there always exists a non-crossing bichromatic
separated matching on at least $(1 + \eps)n$ 
points of $P$. This is a properly colored 
matching whose segments are pairwise disjoint 
and intersected by common line. 
For both versions, this is the 
first improvement of the easily obtained lower 
bound of $n$ by an additive term linear in $n$. 
The best known published upper bounds are 
asymptotically of order $4n/3+o(n)$.
\end{abstract}

\section{Introduction}
\label{sec:introduction}

We study a family of problems that were 
discovered independently in two different (but
essentially equivalent) settings. 
Researchers in discrete
and computational geometry found a geometric 
formulation, while researchers in 
computational biology and stringology
studied circular words. 
Around 1989, 
Erd{\H o}s asked the following 
geometric question~\cite[p.~409]{brasspach}: 
given a set $P$ of $n$ red and $n$ blue points in 
convex position, how many points of $P$ can always be
collected by a non-intersecting polygonal path $\pi$ 
with vertices in $P$ such that the vertex-color
along $\pi$ alternates between red and blue. 
Taking every other segment of $\pi$, we
obtain a properly colored set of pairwise
disjoint segments with endpoints in $P$. A 
closely related problem asks for a large
\emph{separated matching}, a collection 
of such segments with the extra property that all 
of them are intersected by a 
common line. This is equivalent to 
finding a long \emph{antipalindromic subsequence}
in a circular sequence of $2n$ bits,
where $n$ bits are $0$ 
and $n$ bits are $1$, see Figure~\ref{fig:intro}.
This formulation was
stated in 1999 in a paper on protein 
folding~\cite{Lyngso-Pedersen}.
Similar questions were also studied for 
\emph{palindromic}
subsequences~\cite{MullnerR19}. One such
question is 
equivalent to finding many disjoint
monochromatic segments with endpoints in $P$, a problem that was also studied by 
the geometry
community.

\begin{figure}
    \centering
    \includegraphics{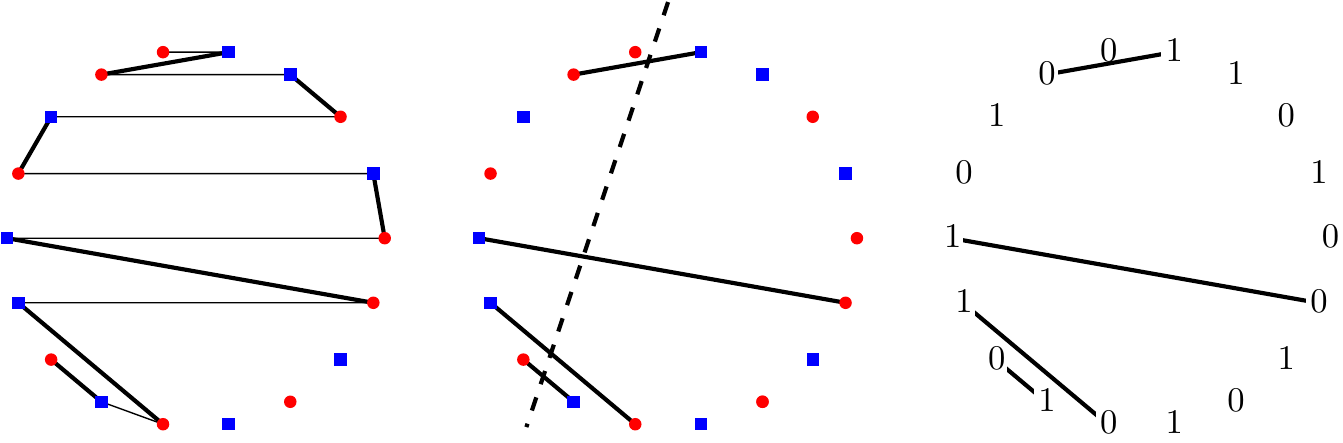}
    \caption{Left: a set $P$ of $18$ points in convex
    position, $9$ of them red and $9$ of them blue,
    with an alternating path of length $15$
    (this is not a longest such path). Taking every
    other segment, we obtain a properly colored
    disjoint matching on $P$. Middle: a separated
    matching on $P$, with a dashed
    line that intersects all matching edges (this
    is not a maximum such matching). Right: 
    an antipalindromic subsequence on a circular
    word of $18$ bits, $9$ of them $0$ and $9$ of
    them $1$.}
    \label{fig:intro}
\end{figure}
An easy lower bound for alternating paths 
is $n$, and the best known lower bound is 
$n + \Omega(\sqrt n)$~\cite{ViolaPhD}. 
We increase this to $cn + o(n)$, 
for a constant $c > 1$. Similarly,
for the other mentioned
problems, we improve the lower bounds
by an additive term of $\eps n$,
for some fixed $\eps > 0$. Also here,
this constitutes the first $\Omega(n)$
improvement over the trivial
lower bounds.

\subsection{The (geometric) setting}
We have a set $P$ of $2n$ points 
$p_0, p_1, \dots, p_{2n - 1}$ in convex
position, numbered in clockwise order. 
The points in $P$ are colored red
and blue, so that  
there are exactly $n$ red points and $n$ 
blue points.
The goal is to find a long \emph{non-crossing 
alternating path} in $P$. That is, a sequence 
$\pi: q_0, q_1, \dots, q_{\ell - 1}$ of points in $P$ such 
that (i) each point from $P$ appears at most once in $\pi$; 
(ii) $\pi$ is \emph{alternating}, i.e., for
$i = 0, \dots, \ell - 2$, we have that
$q_i$ is red and $q_{i+1}$ is
blue or that $q_i$ is blue and
$q_{i + 1}$ is red;
(iii) $\pi$ is \emph{non-crossing}, i.e., for 
$i, j \in \{ 0, \dots, \ell - 2 \}$,
$i \neq j$, the
two segments $q_iq_{i + 1}$ and $q_jq_{j + 1}$
intersect only in their endpoints and only if they are consecutive in $\pi$, see
Figure~\ref{fig:intro}(left).
We will also just say \emph{alternating path} for
$\pi$.
Alternating paths for planar 
point sets in general (not just convex) position 
have been studied in various previous papers, e.g.,~\cite{Abellanas03,AichholzerEtal19,AkiyamaU90,Cibulka2013,Claverol13}.

For most of this work, we will focus on another,
closely related, structure.
A \emph{non-crossing separated bichromatic matching} $M$ in $P$ 
is a set $\{p_1q_1, p_2q_2, \dots, p_k q_k\}$ of $k$ pairs
of points in $P$, such that (i) all points 
$p_1, \dots, p_k, q_1, \dots, q_k$
are pairwise distinct; (ii) the segments $p_iq_i$ and $p_jq_j$ are
disjoint, for all $1 \leq i < j \leq k$; (iii) for $i = 1, \dots, k$,
the points $p_i$ and $q_i$ have different colors; and 
(iv) there exists a line that intersects all
segments $p_1q_1, p_2q_2, \dots, p_kq_k$,
see Figure~\ref{fig:intro}(middle).
Often, we will just use the term
\emph{separated bichromatic matching} or simply \emph{separated matching} for $M$.

\subsection{Previous results}
The following basic lemma says that a large separated
matching immediately yields a long alternating
path. The (very simple) proof was given by
Kyn\v cl, Pach, and T\'oth~\cite[Section 3]{KynclPaTo08}.
\begin{lemma}
\label{lem:matching-alternating}
Suppose that a bichromatic convex 
point set $P$ admits a separated 
matching with $k$ segments.
Then, $P$ has an alternating path of length
$2k$.
\end{lemma}

Let $l(n)$ be the largest number such that 
for every set $P$ of $n$ red and $n$ blue
points in convex position, there is an
alternating path of length at least $l(n)$.
Around 1989, Erd{\H o}s and
others~\cite{KynclPaTo08} conjectured 
that $\lim_{n \rightarrow \infty} l(n)/n = 3/2$.
Abellanas,  Garc\'ia,  Hurtado, and Tejel~\cite{Abellanas03} and, independently,
Kyn\v cl, Pach, and T\'oth~\cite[Section 3]{KynclPaTo08} disproved this  
by showing the upper bound 
$l(n) \leq 4n/3 + O(\sqrt{n})$. 
Kyn\v cl, Pach, and T\'oth~\cite{KynclPaTo08}
also improved the (almost trivial) lower 
bound $l(n) \geq n$ to 
$l(n) \geq n + \Omega(\sqrt{n/\log n})$. 
They conjectured that in fact
$l(n) = 4n/3 + o(n)$. 
In her PhD thesis~\cite{ViolaPhD} (see 
also~\cite{hajnal-viola,Viola11a,Viola11b}), 
M\'esz\'aros improved the lower bound
to 
$l(n) \geq n + \Omega(\sqrt{n})$, and
she described a wide class of 
configurations where every separated 
matching has at most $2n/3 + O(\sqrt{n})$
edges. This also implies the upper bound
$l(n) \leq 4n/3 + O(\sqrt{n})$ mentioned
above~\cite{Abellanas03,KynclPaTo08}.
It was announced to us in personal
communication that E.~Cs\'oka, Z.~Bl\'azsik,
Z.~Kir\'aly, and D.~Lenger constructed
configurations with an upper bound 
of $cn + o(n)$ on the size of the largest
separated matching, where 
$ c =2 - \sqrt{2} \approx 0.5858$.

\subsection{Our results}
We improve the almost trivial lower bound
$n/2$ for separated matchings
to $n/2 + \eps n$.

\begin{theorem}
\label{thm:main}
There is a fixed $\eps > 0$ such that any convex point set 
$P$ with $n$ red and $n$ blue points 
admits
a separated matching with at least
$n/2 + \eps n$ edges.
\end{theorem}

By Lemma~\ref{lem:matching-alternating}, we
obtain the following corollary
about long alternating paths.

\begin{theorem}
\label{thm:main-alternating}
There is a fixed $\eps > 0$ such that
any convex point set $P$ with $n$ red and $n$ blue points 
admits
an alternating path with at least
$n + \eps n$ vertices.
\end{theorem}

A variant of Theorem~\ref{thm:main} 
also holds for the
\emph{monochromatic} case.
The definition of a \emph{non-crossing
separated monochromatic matching}, or simply
\emph{separated monochromatic matching},
is obtained from the definition of a 
separated bichromatic matching 
by changing
 condition (iii) to 
(iii') for $i = 1, \dots, k$,
the points $p_i$ and $q_i$ have the same
color.
Some of the upper bound constructions for
separated bichromatic matchings apply to
the monochromatic setting, also 
giving the upper bound $2n/3+O(\sqrt{n})$.
Here is a monochromatic version of
Theorem~\ref{thm:main}.

\begin{theorem}
\label{thm:main-monochromatic}
There are constants $\eps >0$ and $n_0 \in \N$ such 
any convex point set $P$ with $n \geq n_0$ points, colored
red and blue, admits
a separated monochromatic matching with at
least $n/2 + \eps n$ vertices.
\end{theorem}

There are two differences between
the statement of Theorem~\ref{thm:main}
and Theorem~\ref{thm:main-monochromatic}:
we do not require that the number
of red and blue points in $P$ is equal
(and hence the size of the matching
is stated in
terms of vertices instead 
of edges), and we need a lower bound
on the size of $P$. This is
necessary, because Theorem~\ref{thm:main-monochromatic}
does not always hold for, e.g., $n = 4$.
It was announced to us in a personal
communication that the construction of 
E.~Cs\'oka, Z.~Bl\'azsik, Z.~Kir\'aly and 
D.~Lenger from above also gives the
upper bound $cn + o(n)$ on the size of a
largest separated monochromatic matching,
where $c = 2 - \sqrt{2} \approx 0.5858$.

\subsection{Our results in the setting of
finite words}
\label{sec:biology}

As we already said, the problems 
in this paper were independently 
discovered by researchers in computational
biology and stringology.
In a study on protein folding algorithms,
Lyngs{\o} and Pedersen~\cite{Lyngso-Pedersen}
formulated a conjecture that is equivalent 
to saying that the bound in
Theorem~\ref{thm:main}
can be improved to $2n/3$ (for $n$ divisible
by $3$).
M\"ullner and
Ryzhikov~\cite[p.~461]{MullnerR19} 
write that this conjecture
``has drawn substantial attention from the
combinatorics of words community''.
For the convenience of readers from this
community, we rephrase our theorems
for separated matchings in the finite words
setting. We use the terminology of 
M\"ullner and Ryzhikov~\cite{MullnerR19},
without introducing it here.
The following  corresponds to
Theorem~\ref{thm:main}.

\begin{theorem}
\label{thm:main-words}
There is a fixed $\eps > 0$ such that for 
any even $n \in \N$, every 
binary circular word of length $n$ with 
equal number of zeros and ones has an
antipalindromic subsequence of length at 
least $n/2 + \eps n$.
\end{theorem}

The following corresponds to Theorem~\ref{thm:main-monochromatic}.

\begin{theorem}
\label{thm:main-monochromatic-words}
There are constants $\eps > 0$ and $n_0 \in \N$ 
so that for any $n \in \N$, $n \geq n_0$,
every binary circular
word of length $n$ has a palindromic
subsequence of length at least
$n/2 + \eps n$.
\end{theorem}

\section{Existence of large separated 
bichromatic matchings}
\label{sec:bichromatic}

In this section, we prove our main result: large
separated bichromatic matchings exist.

\subsection{Runs and separated matchings}
\label{sec:proofmanyruns}

A \emph{run} of $P$ is a maximal sequence
$p_i, p_{i + 1}, \dots, p_{i + \ell}$ of consecutive points
with the same color.\footnote{When calculating with indices 
of points in $P$, we will always work modulo $2n$.} That
is, for $j = i, \dots, i + \ell - 1$, the color of $p_j$ and
of $p_{j + 1}$ are the same, and the colors of $p_{i - 1}$ and 
$p_i$ and the colors of $p_{i + \ell}$ and $p_{i + \ell + 1}$
are different.
The number of runs is always even.
Kyn\v cl, Pach, and T\'oth showed that
if $P$ contains $t$ runs, then $P$ admits an
alternating path of length $n +
\Omega(t)$~\cite[Lemma~3.2]{KynclPaTo08}.
We will need the following analogous result for separated 
matchings.

\begin{theorem}
\label{thm:manyruns}
Let $c_1 = 1/32$ and $t \geq 4$.
Let $P$ be a bichromatic convex point set with $2n$ points, $n$ red and $n$ blue,
and suppose that $P$
has $t$ runs. Then, $P$ admits a separated matching with at least
$n/2 + c_1 t^2/n$ edges.
\end{theorem}

\begin{proof}
We partition the edges of the complete geometric 
graph on $P$ into $2n$ parallel matchings $M_0,\dots,M_{2n-1}$, see
Figure~\ref{fig:thm15}.
\begin{figure}
    \centering
    \includegraphics{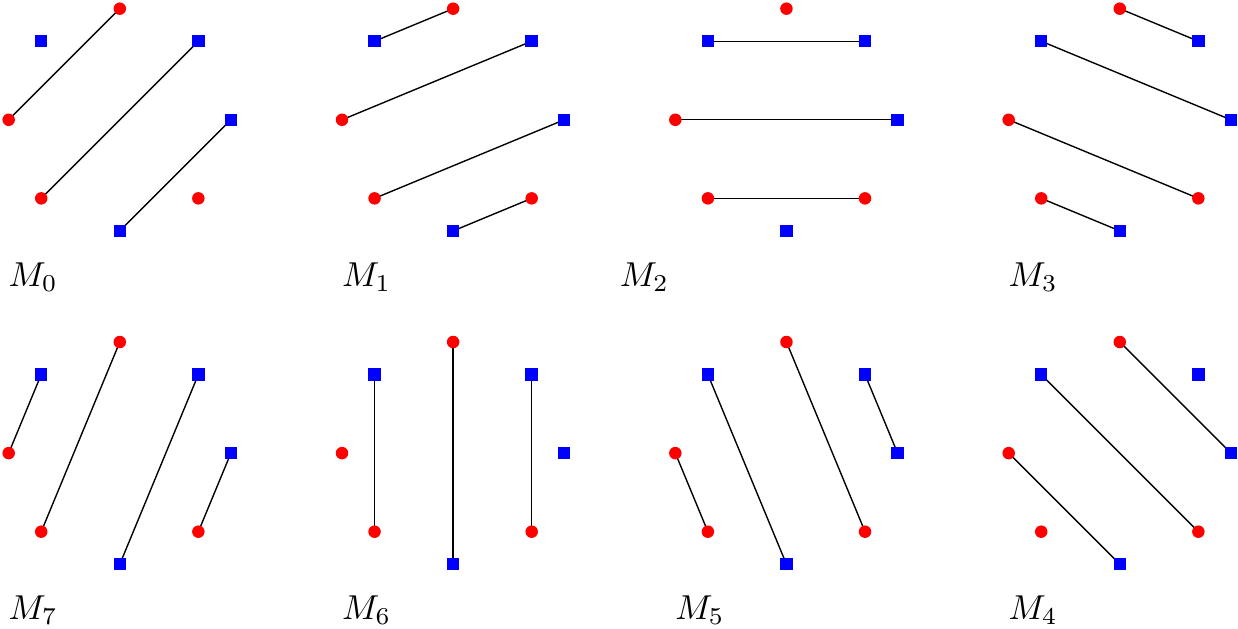}
    \caption{We partition the edges of
    the complete geometric graph on $P$ into
    $2n$ parallel matchings.}
    \label{fig:thm15}
\end{figure}
For $i = 0, \dots, 2n-1$, let $M_i'$ be the 
submatching of $M_i$ that consists of the
bichromatic edges of $M_i$. Every 
$M_i'$ is a separated matching, and the $2n$ 
matchings $M_0',\dots,M_{2n-1}'$ together 
contain all the $n^2$ bichromatic edges on $P$. 
Thus, the average number of edges in
a matching from $M_0', \dots, M_{2n-1}'$ is 
$n^2/2n = n/2$.

Suppose now that $p_j$ and $p_k$ are two 
distinct red points such that $p_j$ is the
(clockwise)
first point of a red run and 
$p_k$ is the (clockwise) 
last point of a different red run. 
Let $M_i$ be the parallel matching that contains
the edge
$p_jp_k$. Then, $p_{j - 1}$ and
$p_{k + 1}$ are blue, and
either $p_{j - 1} = p_{k + 1}$ or
$p_{j-1}p_{k+1}$ is
a monochromatic blue edge in $M_i$.
Thus, the matching 
that is obtained from $M_i'$
by adding the bichromatic edge $p_jp_{k+1}$ is
still a separated matching, and similarly for the
matching $M_i''$ obtained from $M_i'$ by
making all the possible additions of this kind, 
see Figure~\ref{fig:thm15_b}. 
\begin{figure}
    \centering
    \includegraphics{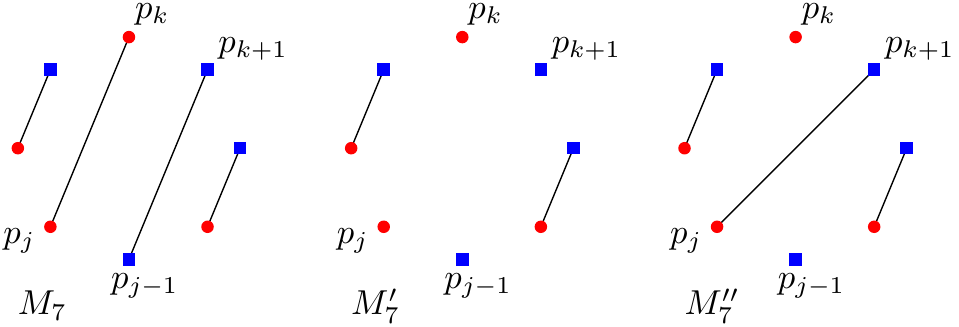}
    \caption{For any monochromotic edge 
    $p_jp_k$ in $M_j$ that connects the first
    (clockwise) point $p_j$ of a red run with the
    last (clockwise) point $p_k$ of another red run,
    we add to $M_j'$ the bichromatic edge 
    $p_jp_{k+1}$. The resulting matching is
    called $M_i''$.}
    \label{fig:thm15_b}
\end{figure}
Since there are $t/2$ red runs, 
the total number of edges that
we add in the matchings
$M_0'', \dots, M_{2n-1}''$ is
$\binom{t/2}{2}$.  Hence, the average size of a
matching from $M_0'', \dots, M_{2n-1}''$ is 
\[
\frac{n}{2}+\frac{\binom{t/2}{2}}{2n}=
\frac{n}{2} + \frac{t(t - 2)}{16n}
\geq \frac{n}{2} + \frac{t^2}{32n},
\]
since $t \geq 4$ and hence $t-2 \geq t/2$.
In particular, at least one of the
matchings $M_i''$ has the desired number
of edges.
\end{proof}

\subsection{Chunks, partitions, and configurations}
\label{sec:chunks}
Let $k \in \{1, \dots, n\}$. A \emph{$k$-chunk} is 
a sequence of consecutive points in $P$ with 
 exactly $k$ points of one color and 
less than $k$ points of the other color. Hence, 
a $k$-chunk has at least $k$ and at most 
$2k - 1$ points. A \emph{clockwise $k$-chunk 
with starting point $p_i$} is the shortest 
$k$-chunk that starts from $p_i$ in clockwise 
order. A \emph{counterclockwise $k$-chunk} with
starting point $p_i$ is defined analogously,
going in the counterclockwise direction. For a 
$k$-chunk $C$, we denote by $r(C)$ the number 
of red points and by $b(C)$  the number
of blue points in $C$. We call $C$ a 
\emph{red chunk} if $r(C) = k$ (and hence 
$b(C) < k$) and a \emph{blue chunk} if $b(C) = k$ 
(and hence $r(C) < k$). The \emph{index} of 
$C$ is $b(C)/k$ for a red chunk and
$r(C)/k$ for a blue chunk. Thus, the index of $C$
lies between $0$ and $(k-1)/k$, and it 
measures how ``mixed'' $C$ is.

Next, let $k \in \{1, \dots, n\}$ and
 $\lambda \in \N \cup \{0\}$. We define a 
\emph{$(k, \lambda)$-partition}. Suppose that 
$k$ is odd. First, we construct a maximum 
sequence $C_0, C_1, \dots$ of clockwise 
disjoint $k$-chunks, as follows: we begin 
with the clockwise $k$-chunk $C_0$ with starting point 
$p_0$, and we let $\ell_0$ be the number of 
points in $C_0$. Next, we take the clockwise $k$-chunk 
$C_1$ with starting point $p_{\ell_0}$, and 
let $\ell_1$ be the number of points in $C_1$. 
After that, we take the clockwise $k$-chunk $C_2$ with 
starting point $p_{\ell_0 + \ell_1}$, and so 
on. We stop once we reach the last $k$-chunk 
that does not overlap with $C_0$. Next, we 
construct a maximum sequence $D_0, D_1, \dots$ 
of \emph{counterclockwise} $(k + 3)$-chunks, 
starting with the point $p_{2n - 1}$, in an 
analogous manner. Let $\lambda'$ be the 
minimum of $\lambda$ and the number of $(k+3)$-chunks 
$D_i$. Now, to obtain the $(k, \lambda)$-partition, 
we take $\lambda'$ counterclockwise $(k + 3)$-chunks 
$D_0, \dots, D_{\lambda' - 1}$ and a maximum number 
of clockwise $k$-chunks $C_0, C_1, 
\dots$ that do 
not overlap with $D_0, \dots, D_{\lambda' - 1}$.
If $k$ is even, the $(k, \lambda)$-partition is 
defined analogously, switching the roles of the 
clockwise and the counterclockwise direction.
\begin{figure}
    \centering
    \includegraphics{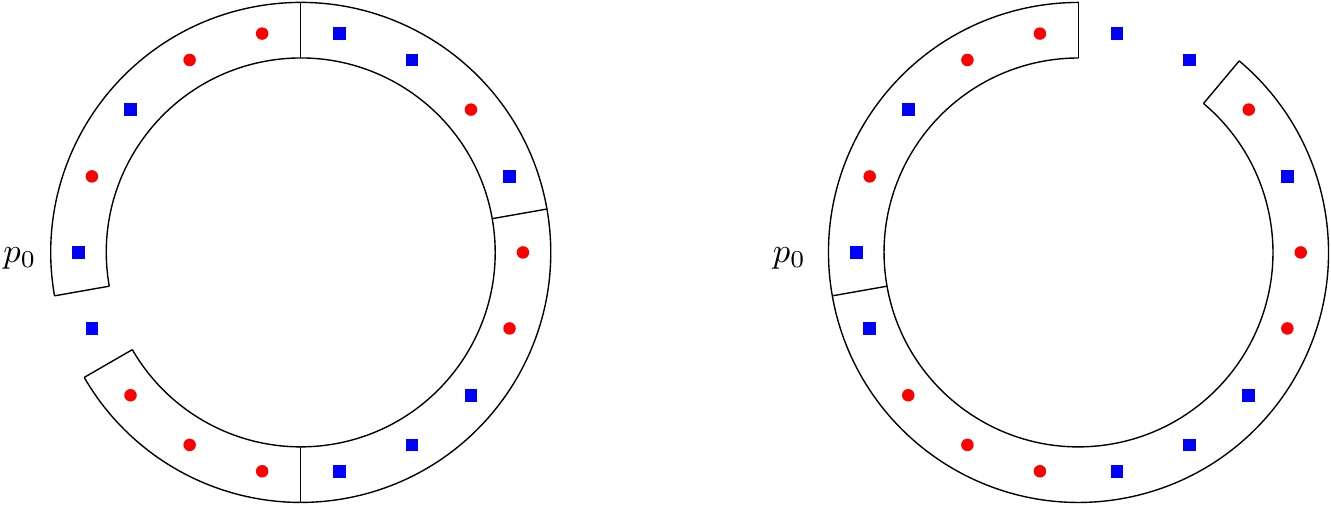}
    \caption{A set of 18 points and its 
    $(3, 0)$-partition (left) and 
    $(3,1)$-partition (right). In the 
    $(3,0)$-partition, the first chunk 
    is red with index $2/3$, the second 
    chunk is blue with index $1/3$,  
    the third chunk is blue with index $2/3$, 
    and the fourth chunk is red with index $0$. 
    The average red index is $1/3$, the 
    average blue index is $1/2$. The index of 
    the $(3, 0)$-partition is $1/2$.
    The $(3, 1)$-partition has one clockwise 
    $3$-chunk and one counterclockwise $6$-chunk.
    The $3$-chunk is red with index $2/3$,
    the $6$-chunk is red with index $5/6$.
    The average red index is $3/4$, the average
    blue index is $0$
    The index of the $(3,1)$-partition is $3/4$.}
    \label{fig:chunks}
\end{figure}
There may
be some points
that do not lie in any chunk
of the $(k, \lambda)$-partition.
We call these points \emph{uncovered}.

The \emph{average red index} of $\Gamma$ is the 
average index in a red chunk of $\Gamma$ 
($0$, if there are no red chunks). 
The \emph{average blue index}
of $\Gamma$ is defined analogously.
The \emph{index} of $\Gamma$ is the maximum
of the average red index and the average
blue index of $\Gamma$. The \emph{max-index color}
is the color whose average index achieves the
index of $\Gamma$, the other color is called
the \emph{min-index color}, see Figure~\ref{fig:chunks} 
for an
illustration of the concepts so far.
The following simple proposition helps
us bound the number of chunks.

\begin{proposition}
\label{prop:chunkcount}
Let $P$ be a convex bichromatic
point set with $2n$ points, $n$ red
and $n$ blue, and 
let $\Gamma$ be a $(k, \lambda)$-partition
of $P$. 
In $\Gamma$, there are at 
most $2k - 2$ uncovered points,
at most $k - 1$ of them red
and at most $k - 1$ of them blue.
Furthermore, let $R$ be the number
of red chunks and $B$ the number of 
blue chunks in $\Gamma$, and let
$\alpha$ be the index of $\Gamma$.
Then,
\begin{equation}
\label{equ:chunkupper}
  R + B \leq \frac{2n}{k} \quad \text{and} \quad
  \max \{R, B\} \leq \frac{n}{k}.
\end{equation}
Furthermore,  we have 
\begin{multline}
\label{equ:chunklower}
 R + B \geq \left\lfloor \frac{2n}{2k + 5}
 \right\rfloor
  > \frac{2n}{7k} - 1,
\quad 
\max \{R, B\} \geq 
\frac{1}{2}\left\lfloor \frac{2n}{2k + 5}
 \right\rfloor
  > \frac{n}{7k} - \frac{1}{2},
  \\\text{and} \quad
  \min \{R, B\} \geq
  \frac{1 - \alpha}{2} \left \lfloor
\frac {2n}{2k + 5} \right \rfloor 
- \frac{k - 1}{k + 3}
> (1-\alpha)\frac{n}{7k} - 2.
\end{multline}
If $\lambda = 0$, the lower bounds improve
to 
\begin{multline}
\label{equ:chunklowerlzero}
R + B \geq \left\lfloor \frac{2n}{2k - 1} \right\rfloor 
\geq \frac{n}{k} - 1,
\quad 
\max \{R, B\} \geq 
\frac{1}{2}\left\lfloor \frac{2n}{2k - 1}
 \right\rfloor
  > \frac{n}{2k} - \frac{1}{2},
  \\\text{and} \quad
  \min \{R, B\} \geq
  \frac{1 - \alpha}{2} \left \lfloor
\frac {2n}{2k - 1} \right \rfloor 
- \frac{k - 1}{k}
> (1 - \alpha)\frac{n}{2k} - 2.
\end{multline}
\end{proposition}

\begin{proof}
Since a red chunk in $\Gamma$ contains
at least $k$ red points, we have 
$R \leq n/k$. Similarly, $B \leq n/k$, and 
(\ref{equ:chunkupper}) follows. 

Any chunk in 
$\Gamma$ has at most $2k + 5$ points.
Thus, the total number of chunks is 
at least 
$R + B \geq \lfloor 2n/(2k + 5 )\rfloor >
2n/7k - 1$, using $k \geq 1$.
This is the first bound of (\ref{equ:chunklower}).
The second bound follows from 
the inequality $\max \{R, B\} \geq (R + B)/2$.
For the third bound, suppose
that $R = \max \{R, B\}$ and $B = \min\{R, B\}$.
Let $\gamma$ be the difference between the
number of uncovered blue points and the
number of uncovered red points in $\Gamma$.
Since the number of red points in $P$ 
is the same as the number of blue points, 
we have
\begin{equation}
\label{equ:balance}
0 = \gamma + \sum_{C} (b(C) - r(C)) +
\sum_{D} (b(D) - r(D)),
\end{equation}
where the first sum ranges over the 
red chunks $C$ of $\Gamma$, and the second
sum ranges over the blue chunks $D$ of
$\Gamma$. Now, we have $\gamma \leq k - 1$,
since there are at most $k - 1$ uncovered blue
points. Furthermore, we have
\[
\sum_C (b(C) - r(C))
= \sum_C r(C) 
\left( \frac{b(C)}{r(C)} - 1\right)
\leq (k + 3) \sum_C 
\left(\frac{b(C)}{r(C)} - 1\right) = 
(k + 3) (\alpha - 1)R,
\]
since $r(C) \leq k + 3$ 
for
every chunk $C$ and 
$(1/R)\sum_C b(C)/r(C) = \alpha$. 
Finally, we have $b(D) - r(D) \leq k  + 3$
for every blue chunk $D$. 
Thus, (\ref{equ:balance})
implies that
\[
0 \leq k - 1 + 
(k + 3) (\alpha - 1)R + (k + 3)B
\quad \Rightarrow \quad 
B \geq (1-\alpha) R - \frac{k - 1}{k + 3}. 
\]
Using the previous lower bound on $R$, this 
gives
\[
B \geq \frac{1 - \alpha}{2} \left \lfloor
\frac {2n}{2k + 5} \right \rfloor 
- \frac{k - 1}{k + 3}
> (1- \alpha)\frac{n}{7 k} - 2,
\]
and hence~(\ref{equ:chunklower}). 
To obtain~(\ref{equ:chunklowerlzero}),
we argue in the same manner, but we
use the fact that for $\lambda = 0$,
all chunks have at most 
$2k - 1$ points and exactly $k$ points
of the majority color.
\end{proof}

The purpose of the $(k, \lambda)$-partitions is 
to transition smoothly between the $(k, 0)$-partition 
and the $(k + 3, 0)$-partition. In our proof, this will 
enable us to gradually increase the chunk-sizes,
while keeping the index 
under control.

A \emph{$k$-configuration} of $P$ is a partition of
$P$ into $k$-chunks, leaving no uncovered
points, see Figure~\ref{fig:configuration}. 
In contrast
to a $(k, \lambda)$-partition, the chunks in a
$k$-configuration are not necessarily minimal.
Note that while $P$ always has a $(k, \lambda)$-partition,
it does not necessarily 
admit a $k$-configuration. The average
red index, the average blue index, etc.~of 
a $k$-configuration
are defined as for a $(k, \lambda)$-partition.
The following proposition helps us bound
the number of chunks in a $k$-configuration.

\begin{figure}
    \centering
    \includegraphics{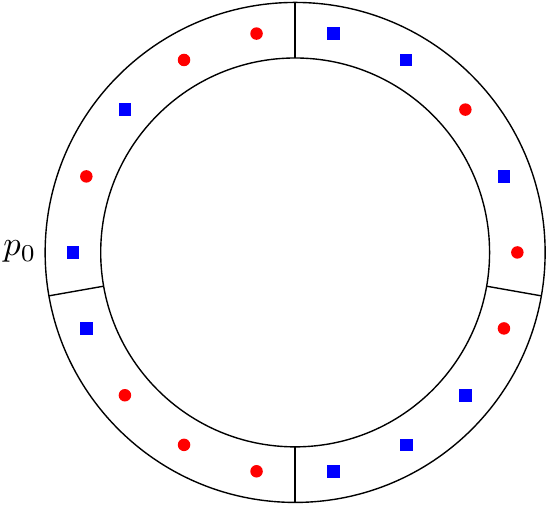}
    \caption{A set of 18 points and a $3$-configuration for
    it. The chunk from $p_0$ is red with index $2/3$, the
    next clockwise chunk is blue with index $2/3$,
    followed by another blue chunk of index $1/3$
    and a final red chunk of index $1/3$. The average blue
    index and the average red index are both $1/2$. Note that
    the chunks are not minimal.}
    \label{fig:configuration}
\end{figure}

\begin{proposition}
\label{prop:confcount}
Let $P$ be a convex bichromatic
point set with $2n$ points, $n$ red
and $n$ blue, and 
let $\Gamma$ be a $k$-configuration
of $P$. 
Let $R$ be the number
of red chunks, $B$ the number of 
blue chunks,$\alpha$ the 
average red index and $\beta$
the average blue index of $\Gamma$.
Then, 
\begin{equation}
\label{equ:chunkpoints}
n = kR + \beta kB = kB + \alpha  k R.
\end{equation}
Furthermore, 
$R + B \geq n/k$, $\max\{R, B\} \geq n/2k$, and
$\min\{R, B\} \geq (1  - \max\{\alpha, \beta\})n/2k$.
Finally, $\max\{R, B\} = R$ if 
and only if $\alpha \geq \beta$.
\end{proposition}

\begin{proof}
The red chunks in $\Gamma$ contain 
$kR$ red points and $\alpha kR$ blue points,
while
the blue chunks contain $kB$ blue points
and $\beta kB$ red points. All points
are covered by the chunks, and there 
are $n$ red points and $n$ blue points. 
This implies
(\ref{equ:chunkpoints}).
Since all points in $P$ are covered
by the chunks in $\Gamma$, the lower bound
$R + B \geq \lfloor 2n/(2k - 1) \rfloor$ 
from (\ref{equ:chunklowerlzero})  becomes
$R + B \geq 2n/(2k - 1) \geq n/k$, and
hence also $\max\{R, B \} \geq n/2k$.
Finally, from (\ref{equ:chunkpoints}),
we get
$(1- \alpha)R = (1 - \beta)B$,
so $\max \{R, B\} = R$ if and only if
$\alpha \geq \beta$, i.e., the number
of chunks of the max-index color
is at least the number of chunks of the
min-index color.
This also implies that 
$B \geq (1 - \alpha)R$ and 
$R \geq (1 - \beta )B$.
In particular,
$\min\{R, B\} \geq (1 - \max\{\alpha, \beta))
\max\{R, B\}$.
\end{proof}

In our proof, the key challenge will be 
to analyze $k$-configurations with 
small  constant index (say, 
around $0.1$).

\subsection{From 
\texorpdfstring{$(k, \lambda)$}{(k, lambda)}-partitions to
\texorpdfstring{$k$}{k}-configurations}

Our first goal is to show that we can focus on
$(k, \lambda)$-partitions with large $k$ and constant,
but not too large index.  We begin by noting that
if the $(k, 0)$-partition of $P$ for a constant $k$ has a large
index, then we can find a long alternating path in 
$P$. 

\begin{lemma}
\label{lem:smallchunks}
Set $c_2 = 1/12800$. 
Let $k, n \in \N$ with
$8k^2 \leq n$. Let $P$ be
a convex bichromatic point set with
$2n$ points, $n$ red and $n$ blue.
If the
$(k, 0)$-partition $\Gamma$ of $P$ has index at least $0.1$, 
then $P$ admits a separated matching of size at least $(1/2 + c_2/k^4)n$.
\end{lemma}

\begin{proof}
If $P$ has only two runs, then $P$ has a 
separated matching of size $n$, and the
theorem follows since $n/2 \geq c_2 (n/k^4)$.
Thus, assume that $P$ has at
least $4$ runs, and suppose for concreteness that the
max-index color in $\Gamma$ is
red. Let $R$ be the number of 
red chunks in $\Gamma$.
Since the
index of $\Gamma$ is
at most $(k-1)/k$, by~(\ref{equ:chunklowerlzero}),
we have 
\[
R \geq 
\frac{1}{k}  \cdot 
\frac{n}{2k} - 2 
= 
\frac{n - 4k^2}{2k^2}
\geq 
\frac{n - n/2}{k^2}
= \frac{n}{2k^2},
\]
since $n \geq 8k^2$ and hence $4k^2 \leq n/2$.
The average red index is at least $0.1$, 
and the index of a chunk is at most $1$.
Thus, there must be at least 
$0.1R \geq 0.05 (n/k^2)$ red chunks with 
positive index and hence at least 
$0.05 (n/k^2)$ chunks with at least one red point 
and one blue point.
In particular, 
$P$ has at least
$0.05(n/k^2)$ runs. We also
assumed that $P$ has at least four
runs.
Thus, by Theorem~\ref{thm:manyruns}, it follows that $P$
admits a separated matching of size 
\[
\frac{n}{2} + \frac{n^2}{400 k^4 \cdot 32 n} = 
\left(\frac{1}{2} + \frac{c_2}{k^4}\right)n.
\qedhere
\]
\end{proof}

Next, we show that if the $(k, 0)$-partition still
has a small index for 
$k = \Omega(n)$, then
we can find a large separated matching.

\begin{lemma}
\label{lem:largechunks}
Set $c_3 = 1/81$.
Let $k, n \in \N$ with
$k \leq n$ and $6480 n  \leq k^2$.
Let $P$ be
a convex bichromatic point set with
$2n$ points.
If the 
$(k, 0)$-partition 
$\Gamma$ of $P$ has
index at most  $0.1$, then $P$ admits
a separated matching of size at least $(1/2 + 
c_3 (k/n)^2)n$.
\end{lemma}

\begin{proof}
We adapt an argument by Kyn\v{c}l, Pach, and T\'oth~\cite[Lemma~3.1]{KynclPaTo08}.
Set $i_0 = \lceil \log(2n/k) \rceil$.
Then 
\[
k \in \left[\left\lceil \frac{2n}{2^{i_0}} \right\rceil,
\left\lceil \frac{2n}{2^{{i_0}}} \cdot 2 \right\rceil \right).
\]
Since the index of $\Gamma$ is 
at most $0.1$, there is a 
$k$-chunk 
$C'$ in 
$\Gamma$ with index at most $0.1$.
For concreteness,
assume that $C'$ is red. 
We truncate $C'$ to the
first $\lceil 2n/2^{i_0} \rceil$ elements in 
clockwise direction, 
calling the resulting
interval $C$. Since $C'$ has index at most
$0.1$, it follows that $C$ contains 
at most 
\[ 
0.1 \cdot k 
< 0.1 \cdot \left\lceil \frac{2n}{2^{i_0}}
 \cdot 2  \right\rceil
\leq 0.1 \cdot 2 \cdot
\left \lceil \frac{2n}{2^{i_0}}\right \rceil = 0.2|C|
\]
blue points and hence at least 
$0.8|C|$ red points.
We now define a sequence $D_{0}, D_{1}, \dots, D_{i_0 - 1}$ of 
pairwise disjoint intervals.
The intervals are consecutive in clockwise
order after $C$: first come the points from $C$,
then  the points from $D_{0}$, then from $D_1$, etc.
The number of points in $D_i$ is chosen to be  
\[
|D_i| = \left\lceil \frac{2n}{2^{i_0}} \cdot {2^{i + 1}} \right\rceil  - 
\left\lceil \frac{2n}{2^{i_0}} \cdot{2^i} \right\rceil.
\]
Furthermore, for $i = 0, \dots, i_0 - 1$, we set
$F_i = C \cup \bigcup_{j = 0}^{i} D_j$. 
By construction, we have
\[
|F_i| = \left\lceil \frac{2n}{2^{i_0}} \cdot 2^{i+1} \right\rceil,
\]
for  
$i = 0, \dots, i_0 - 1$.
In particular, $F_{i_0 - 1} = P$,
see Figure~\ref{fig:kynclintervals}. Note that the
size of $F_i$ and the size of $D_{i+1}$ differ
by at most $1$, where always $|F_i| \geq |D_{i + 1}|$.
\begin{figure}
    \centering
    \includegraphics{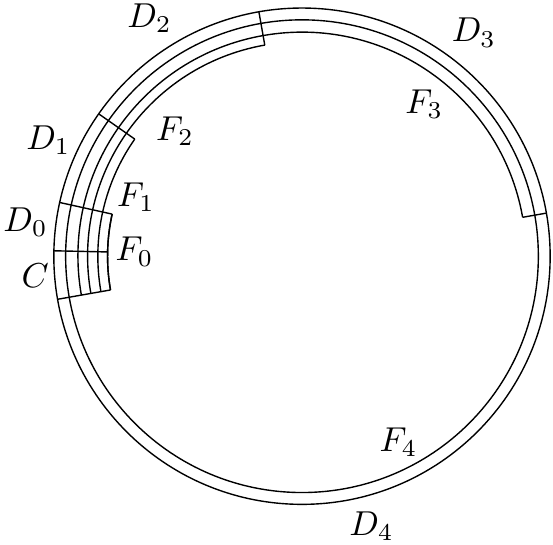}
    \caption{The partition of the convex point
    set into intervals. The size of the intervals
    roughly doubles in each step. The size of
    $F_i$ is roughly the same as the size of
    the following $D_{i+1}$.}
    \label{fig:kynclintervals}
\end{figure}
Now, set $\delta = k/n$ and let $i^*$ be the smallest $0 \leq i \leq i_0 -  1$ such that  
$D_{i}$ contains at least 
\[
 \left( \frac{1}{2} + \frac{\delta}{20} \right)|D_i|
 \]
 blue points, and set $i^* = i_0$,
if no such $i$ exists.
Suppose that $i^* \geq 1$. Then, for $j = 0 \dots, i^* - 1$, the interval
$D_j$ contains at least 
\[
 \left(\frac{1}{2} - \frac{\delta}{20} \right)|D_j|
\]
red points, and hence the total number of red points in the interval $F_{i^* - 1}$
is at least
\begin{align*}
&\phantom{=}\frac{4}{5}|C| + 
\left(\frac{1}{2} - \frac{\delta}{20}\right)  
\sum_{j =  0}^{i^* - 1} |D_j|  \\
&=\frac{4}{5}\left \lceil \frac{2n}{2^{i_0}} \right \rceil + 
\left(\frac{1}{2} - \frac{\delta}{20}\right)  
\sum_{j =  0}^{i^* - 1} \left(\left\lceil \frac{2n}{2^{i_0}} \cdot {2^{j + 1}} \right\rceil  - 
\left\lceil \frac{2n}{2^{i_0}} \cdot{2^j} \right\rceil\right).  \\
\intertext{The sum telescopes, so this is}
&=\frac{4}{5}\left \lceil \frac{2n}{2^{i_0}} \right \rceil + 
\left(\frac{1}{2} - \frac{\delta}{20}\right)  
\left(\left\lceil \frac{2n}{2^{i_0}} \cdot {2^{i^*}} \right\rceil  - 
\left\lceil \frac{2n}{2^{i_0}}\right\rceil\right)  \\
&\geq 
\frac{1}{2} \left \lceil \frac{2n}{2^{i_0}} \cdot 2^{i^*} \right \rceil +
\frac{3}{10} \left \lceil \frac{2n}{2^{i_0}} \right \rceil  - \frac{\delta}{20} 
\left \lceil \frac{2n}{2^{i_0}} \cdot 2^{i^*} \right \rceil.  \\
\intertext{Since for $a > 0$ and $b \in \N$, we
have $b\lceil a \rceil \geq \lceil b a\rceil$, and
hence $\lceil a \rceil \geq (1/b) \lceil a b \rceil$,
this is}
&\geq 
\frac{1}{2} \left \lceil \frac{2n}{2^{i_0}} \cdot 2^{i^*} \right \rceil +
\frac{3}{10} \cdot \frac{1}{2^{i^*}} \cdot \left \lceil \frac{2n}{2^{i_0}} \cdot {2^{i^*}} \right \rceil  - \frac{\delta}{20} 
\left \lceil \frac{2n}{2^{i_0}} \cdot 2^{i^*} \right \rceil.  \\
\intertext{Using that $i^* \leq i_0 = \lceil \log(2n/k)\rceil < \log(2n/k) + 1
= - \log \delta + 2$
and thus $2^{-i^*} > 2^{\log \delta - 2} = \delta/4$,
we lower bound this as}
&\geq 
\left(\frac{1}{2}  +
\frac{3}{10} \cdot \frac{\delta}{4} - \frac{\delta}{20} \right)
\left \lceil \frac{2n}{2^{i_0}} \cdot 2^{i^*} \right \rceil  \\
&=
\left(\frac{1}{2}  +
\frac{\delta}{40} \right)|F_{i^* - 1}|.
\end{align*}
It follows that $i^* \leq i_0 -1$, since $F_{i_0 - 1} = P$, and
$P$ contains only $n = (1/2)|F_{i_0 - 1}|$ red points.

Now, if $i^* = 0$, we set $J_1 = C$ and $J_2 = D_0$.
If $i^* \geq 1$, we set $J_1 = F_{i^*-1}$ and 
$J_2 = D_{i^*}$.
In this way, we obtain two adjacent intervals $J_1$ 
and $J_2$ ($J_2$ clockwise from $J_1$) such that 
the following holds:
if we write $\ell_1 = |J_1|$ and $\ell_2 = |J_2|$,
then (i) 
$\ell_1 \geq |C| > k /2$ and $\ell_2 \geq \ell_1 - 1$; and
(ii) $J_1$ contains at least $(1/2 + \delta/40)\ell_2$ 
red points and $J_2$  contains at least $(1/2 + \delta/40)\ell_2$ 
blue points. 
We match the first $(1/2 + \delta/40)\ell_2$ red points
in $J_1$, counterclockwise from the common boundary of $J_1$ and $J_2$, to the first 
$(1/2 + \delta/40)\ell_2$ blue points  in $J_2$, clockwise from the
common boundary of $J_1$ and $J_2$. 
Let $K \subseteq J_1 \cup J_2$ be the smallest interval that contains the matched
edges, and let $L$ be the complementary interval of $P$.

\begin{figure}
    \centering
    \includegraphics{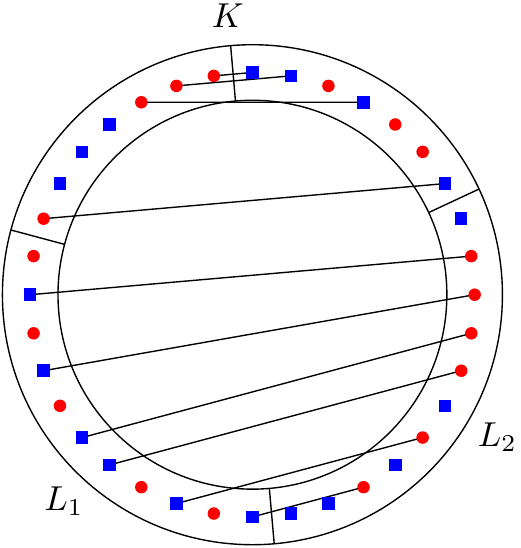}
    \caption{We match the red points in $J_1$ to the blue points
    in $J_2$ and find a large separated matching between the
    remaining points.}
    \label{fig:kynclmatching}
\end{figure}
The interval $K$ contains exactly 
$(1/2 + \delta/40)\ell_2$ red
points from $J_1$ and at most $(1/2 - \delta/40)\ell_2$
red points from $J_2$. Thus, the number of red points
in $K$ is at most $\ell_2$. Similarly, the interval
$K$ contains exactly $(1/2 + \delta/40)\ell_2$ blue
points from $J_2$ and at most $\ell_1 - (1/2 + \delta/40)\ell_2 \leq 1 + (1/2 - \delta/40)\ell_2$
blue points from  $J_1$, so the number of blue points
in $K$ is at most $\ell_2 + 1$. Setting
$m = n - \ell_2 -1$, it follows that
$L$ contains at least
$m$ red points and 
at least $m$ blue points, in particular, 
$|L| \geq 2m$. 

We partition 
$L$ into two intervals $L_1$ and $L_2$, each of size 
at least $m$. Clearly, $L_1$  contains
at least $m/2$ red points or $m/2$ blue points
(or both), and similarly for $L_2$.
Suppose that $L_1$ contains at least $m/2$
red points. If $L_1$ has less than $m/2$ blue points,
then $L_2$ must have at least $m/2$ blue points.
If $L_1$ has at least $m/2$ blue points,
then $L_1$ has at least  $m/2$ red points and
at least $m/2$ blue points, and $L_2$ has at
least $m/2$ red points or $m/2$ blue points.
In any case, it must be that 
$L_1$ contains at least $m/2$
red points and  $L_2$ contains at least $m/2$ blue points,
or vice versa. Thus, we can obtain a bichromatic matching 
between $L_1$ and $L_2$ of size 
at least $m/2$, see Figure~\ref{fig:kynclmatching}. Overall,
we get a separated matching of size at least
\begin{align*}
\left(\frac{1}{2} + \frac{\delta}{40}\right)\ell_2 + \frac{1}{2}\left(n - \ell_2 - 1\right)
&= \frac{n}{2} + \frac{\delta}{40} \cdot \ell_2 - \frac{1}{2}\\
&\geq \frac{n}{2} + \frac{\delta k}{80} - \frac{\delta}{40} - \frac{1}{2}\\
&\geq \frac{n}{2} + \frac{\delta k }{80} - 1\\
&\geq \frac{n}{2} + \frac{\delta^2n}{81},
\end{align*}
since $\ell_2 \geq \ell_1 - 1
\geq k/2-1$ and since 
\[ 
\frac{\delta k}{80} - 1
= \frac{\delta^2  n - 80}{80}
\geq \frac{\delta^2 n - (1/81)\delta^2 n}{80}
= \frac{\delta_1^2 n}{81},
\]
as $6480 n \leq k^2$, so
$80 \leq k^2/(81 n) = \delta^2 n/81$.
\end{proof}

Our goal now is to show that we can focus
on $k$-configurations with $k$ neither too small nor too large, and of index approximately $0.1$.
Here, we only sketch the argument, and we will
make it more precise below, once all the
lemmas have been stated formally: we choose $k_1 = O(1)$ and 
$k_2 = \Omega(n)$
to satisfy the previous two lemmas,
and we consider the
sequence of the $(k_1, 0)$-partition, the $(k_1, 1)$-partition,
the $(k_1, 2)$-partition, $\ldots$, up to the $(k_2, 0)$-partition
 of $P$. By Lemma~\ref{lem:smallchunks} and Lemma~\ref{lem:largechunks},
we can assume that the first partition in the sequence has index
less than $0.1$ and the last partition in the sequence has index
larger than $0.1$. Thus, at some point the index
has to jump over $0.1$. Our definition of $(k, \lambda)$-partition 
ensures that this jump is gradual.

\begin{lemma}
\label{lem:smallchange}
Let $k, n \in \N$ with $n \geq 210000k$. 
Let $P$ be a convex bichromatic point set with
$2n$ points, $n$ red and $n$ blue.
Let $\Gamma_1$ be the $(k, \lambda)$-partition
and $\Gamma_2$  the $(k, \lambda + 1)$-partition
of $P$. 
Suppose that the index of $\Gamma_1$ is 
at most $0.1$.
Then, the average
red index and the average blue index of $\Gamma_1$ and 
$\Gamma_2$ each differ by at most
$0.001$.
\end{lemma}

\begin{proof}
We bound the change of the average red
index, the argument for the average blue
index is analogous.
Let $R$ be the
number of red chunks in $\Gamma_1$,
$R'$ the number of red chunks in $\Gamma_2$, 
and let
$\alpha(C)$ denote the index of a red chunk $C$ in
$\Gamma_1$ or $\Gamma_2$.
We would like to estimate the change of the average
red index of from $\Gamma_1$ to $\Gamma_2$, i.e.,
\begin{equation}
\label{equ:indexchange}
\left| \frac{1}{R'} \sum_{C'} \alpha(C') - 
\frac{1}{R} \sum_{C} \alpha(C)\right|,
\end{equation}
where the first sum goes over all red chunks $C'$ in
$\Gamma_2$ and the second sum goes over all red chunks
$C$ in $\Gamma_1$. We have
\begin{align*}
(\ref{equ:indexchange})
&=
 \frac{\left |R\sum_{C'} \alpha(C') - 
R' \sum_{C} \alpha(C) \right |}{R \cdot R'}\\
&\leq \frac{|R - R'|}{R \cdot R'} \sum_{C_1} \alpha(C_1)
+ \frac{R}{R \cdot R'} \sum_{C_2} \alpha(C_2)
+ \frac{R'}{R \cdot R'} \sum_{C_3} \alpha(C_3),
\end{align*}
where the first sum goes over all red chunks
$C_1$ that appear in both $\Gamma_1$ and $\Gamma_2$,
the second sum goes over all red chunks $C_2$ that
appear only in $\Gamma_2$, and the third sum goes over
all red chunks $C_3$ that appear only in $\Gamma_1$.
When going from $\Gamma_1$ to  $\Gamma_2$,
we add one $(k + 3)$-chunk
and remove the $k$-chunks that overlap with it. A
$(k + 3)$-chunk has at most $2k + 5$ points, and a $k$-chunk
has at least $k \geq 1$ points. Thus, the new 
$(k + 3)$- chunk
can overlap at most seven $k$-chunks. 
This implies that $|R - R'| \leq 6$.
All indices are in $[0, 1)$.
Thus, we have $(1/R') \sum_{C_1} \alpha(C_1) \leq (1/R')
\sum_{C'} \alpha(C') \leq 1$. Moreover, 
since $\Gamma_2$ contains at most one new red
chunk, we have $\sum_{C_2} \alpha(C_2) \leq 1$, and
since $\Gamma_1$ contains at most $7$ 
red chunks that do not
appear in $\Gamma_2$, we have 
$\sum_{C_3} \alpha(C_3) \leq 7$. Thus,
\[
(\ref{equ:indexchange})
\leq \frac{6}{R} + \frac{1}{R'} + \frac{7}{R}.
\]
By (\ref{equ:chunklower}),
we have 
\[
R \geq \frac{(1 - 0.1) n}{7k} - 2 = 
\frac{9n - 140k}{70k} \geq \frac{7n}{70k}
= \frac{n}{10k},
\]
as $n \geq 210000k \geq 70k$.
In particular, $R \geq n/(10k) \geq 8$,
so $R' \geq R - 7 \geq R/8$.
Thus,
\[
(\ref{equ:indexchange})
\leq \frac{6}{R} + \frac{8}{R} + \frac{7}{R}
= \frac{21}{R} \leq \frac{210 k}{n}
\leq 0.001.
\qedhere
\]
\end{proof}

It follows that we can assume that we are dealing
with a $(k, \lambda)$-partition of 
index approximately $0.1$. Actually,
we will see that it suffices to
consider \emph{$k$-configurations}
of index $0.1$. This will be the focus of the
next section.

\subsection{Random chunk-matchings in \texorpdfstring{$k$}{k}-configurations}
\label{sec:chunkmatching}

In this section, we will focus on convex 
bichromatic point sets $P$ that admit a $k$-configuration $\Gamma$
with special properties.
Later, we will see how to reduce to this case.

Let $C_0, C_1, \dots, C_{\ell - 1}$ be the chunks 
of the $k$-configuration $\Gamma$.
We define a notion of \emph{chunk-matching}, 
as illustrated in Figures~\ref{fig:chunkmatching} 
and~\ref{fig:chunkmatching2}.
\begin{figure}
    \centering
    \includegraphics[scale=0.8]{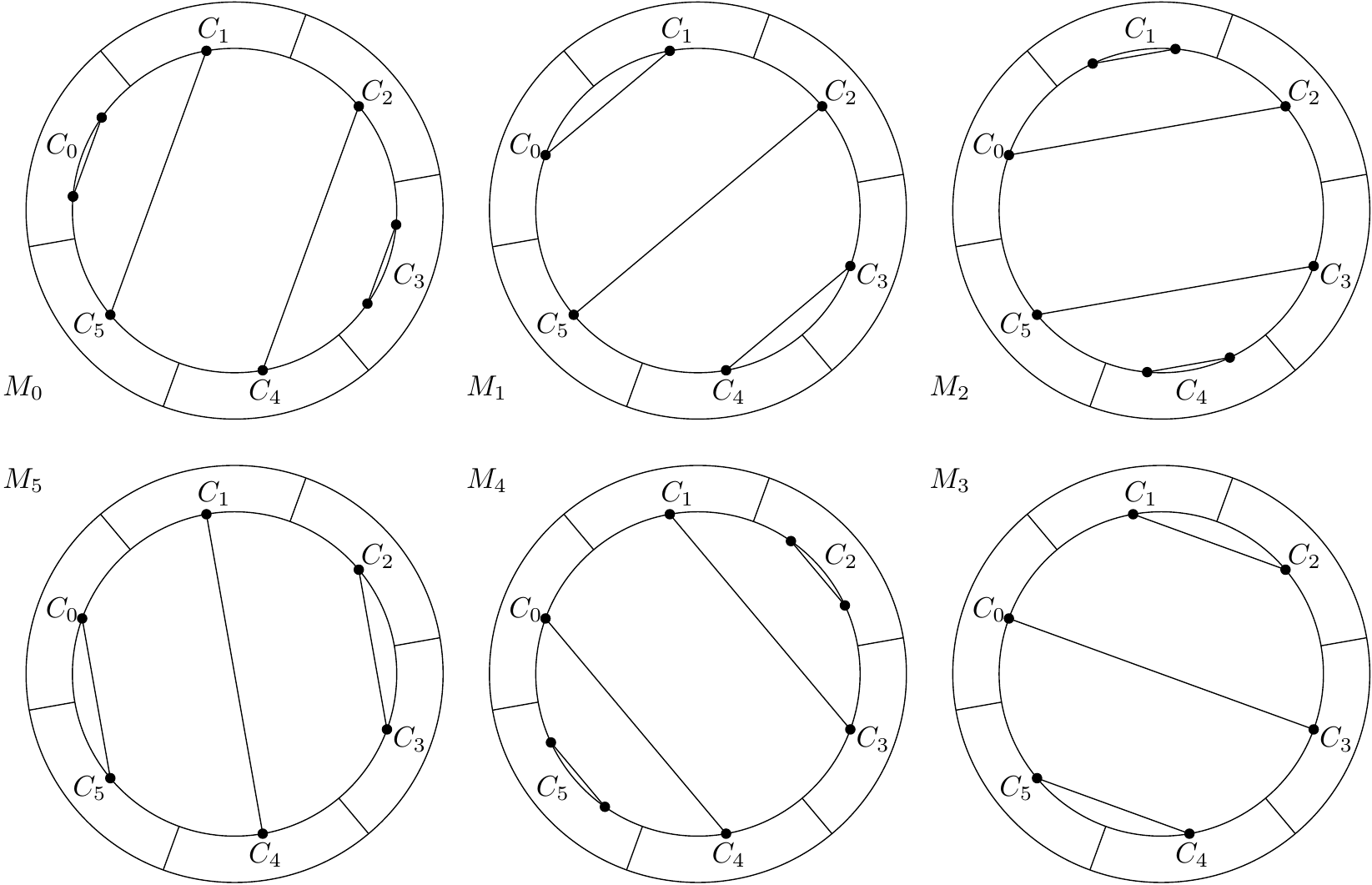}
    \caption{The six chunk matchings $M_0,  \dots,
    M_5$
    for a set of six chunks. If $i$ is even, the chunks 
    $C_{i/2}$ and $C_{i/2 + 3}$ are matched to themselves.
    If $i$ is odd, every chunk is matched to a different
    chunk.}
    \label{fig:chunkmatching}
\end{figure}
\begin{figure}
    \centering
    \includegraphics[scale=0.8]{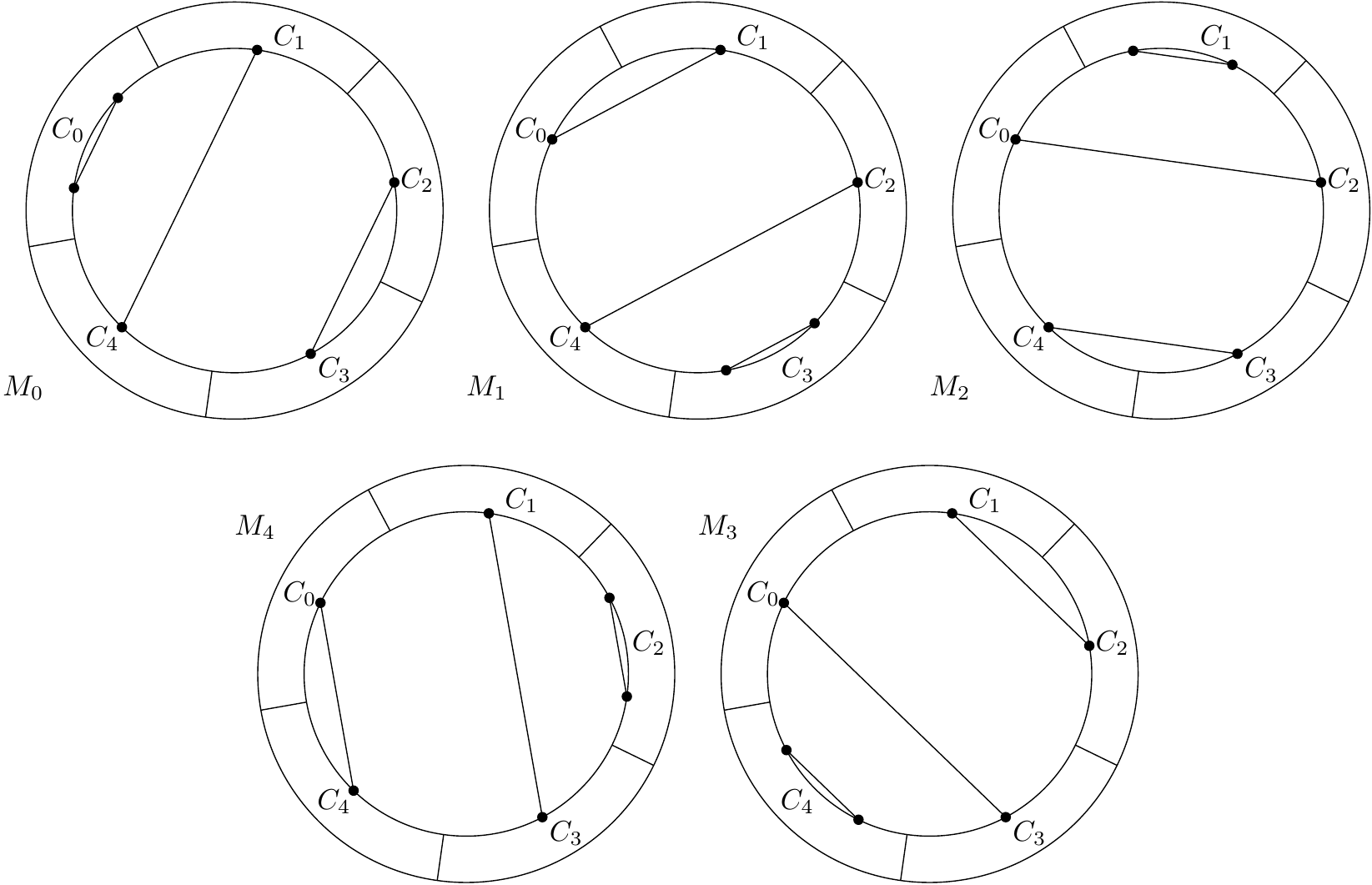}
    \caption{The five chunk matchings $M_0, M_1, \dots,
    M_4$ for a set of five chunks. In matching $M_i$,
    the chunk $C_{(3i \bmod 5)}$ is matched to itself. Every other chunk
    is matched to a different chunk.}
    \label{fig:chunkmatching2}
\end{figure}
A chunk matching pairs each of the $\ell$ chunks
with another chunk (possibly itself). Our goal
is to define chunk matchings in such a way
that we can easily derive from a chunk matching a separated
matching between the points in $P$.

Formally, we define $\ell$ matchings 
$M_0, \dots, M_{\ell - 1}$ by saying that
for $i, j = 0, \dots, \ell - 1$, the
matching $M_i$ pairs the chunks
$C_j$ and $C_{(i-j) \bmod \ell}$.
Again, refer to Figures~\ref{fig:chunkmatching}
and~\ref{fig:chunkmatching2} for examples.
The matching rule is symmetric, i.e., 
if $C_a$ is matched to $C_b$ then 
$C_b$ is matched to $C_a$.
Note that if
$j \equiv (i - j) \pmod \ell$,
the chunk $C_j$ is matched to itself
in $M_i$. If $\ell$ is even, this
happens only for even $i$, namely for
$j = i/2$ and for $j = i/2 + \ell/2$.
If $\ell$ is odd, this happens in every
matching, namely for $j \equiv (\ell + 1)i/2 \pmod \ell$.
By construction, for every $M_i$,
if we connect the matched chunks by straight 
line edges, we obtain a set of plane segments such that there is 
one line that intersects all segments.
Furthermore, every pair $C_i, C_j$ of 
chunks, $0 \leq i \leq j \leq \ell -1$ appears 
in exactly one chunk matching. In essence,
these matchings correspond to partitioning
the chunks of $\Gamma$ with a line,
where the line can possibly pass through
one or two chunks of $\Gamma$ that are
then matched to themselves.

Next, we describe how to
derive from a given chunk matching $M$ a separated
matching on $P$, see Figure~\ref{fig:chunkpairing}
for an illustration. 
We look at every two chunks $C$ and $D$
paired my $M$ (possibly, $C$ = $D$).
If $C$ is red and $D$ blue, we 
match the $k$ red points 
in $C$ to the $k$ blue points in $D$, getting $k$ 
matched edges. The case that $C$ is blue and $D$ is 
red is analogous. If $C \neq D$ and both $C$ and $D$ are red,
we could  match the $k$ red points in $C$ to the 
$b(D) < k$ blue points in $D$, or vice versa. 
We choose the option that gives more edges, 
yielding $\max\{ b(C), b(D)\}$ matched edges.
The case that $C \neq D$ and both  are blue 
is similar. Finally, suppose that $C = D$, 
and for concreteness, 
suppose that $C$ is red. In this case, we split 
the points in $C$ into two parts, containing 
$\lceil k/2 \rceil$ red points each (if $k$ 
is odd, the median point belongs to both parts).
In one part, we have at least $\lceil b(C)/2 \rceil$ 
blue points, and we match these blue points to the 
red points in the other part. This yields
$\lceil b(C)/2 \rceil \geq b(C)/2$ matched edges.
Thus, a chunk matching $M$ gives a separated matching with
at least
\begin{multline}
\label{equ:randomchunkmatching}
\frac{1}{2}
\Bigg(\sum_{\substack{(C, D) \in M \\ C\text{ red}, D \text{ red}}} \max \{b(C), b(D)\}
+
\sum_{\substack{(C, D) \in M \\ C\text{ red}, D \text{ blue}}} k\\
+
\sum_{\substack{(C, D) \in M \\ C\text{ blue}, D \text{ red}}} k
+ \sum_{\substack{(C, D) \in M \\ C\text{ blue}, D \text{ blue}}} \max \{r(C), r(D)\}
\Bigg)
\end{multline}
matched edges, where the sums go over all ordered pairs of matched chunks in $M$,
i.e., a matched pair $(C, D)$ with $C \neq D$ appears twice (which is compensated
by the leading factor of $1/2$) and a matched 
pair $(C, C)$ appears once. The next lemma shows that
a chunk matching that is chosen uniformly at random usually matches half the points of $P$.
\begin{figure}
    \centering
    \includegraphics{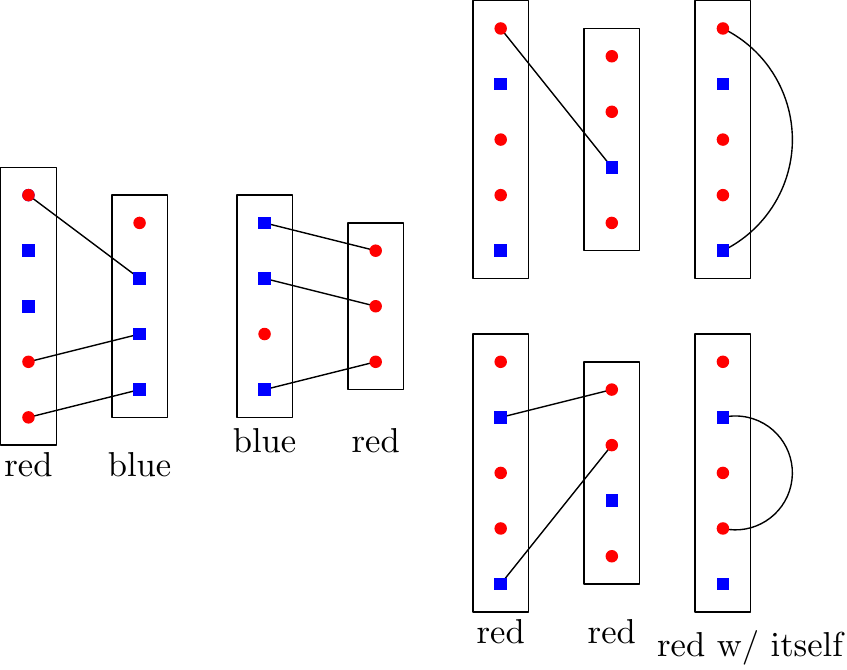}
    \caption{Going from a matched pair of
    chunks to a separated matching.
    If the two chunks have different colors,
    we can match $k$ edges. If the two
    colors are the same, there are two 
    reasonable options, matching the red 
    points in one chunk with the blue 
    points in the other chunk. We choose
    the one that matches more edges. A special
    case occurs if a chunk is matched to itself.
    In this case, we split the majority
    color into half and match between
    the halves.}
    \label{fig:chunkpairing}
\end{figure}
\begin{lemma}
\label{lem:randomchunkmatch}
Let $\Gamma$ be a $k$-configuration of $P$ and $M$
a random chunk matching in $\Gamma$.
The expected number of matched edges in 
the corresponding separated matching is at least 
$n/2$.
\end{lemma}

\begin{proof}
Let $R$ be the number of red chunks in $\Gamma$ and
$B$  the number of blue chunks in $\Gamma$.
Let $\alpha$ be the average index of the red chunks,
and $\beta$ the average index of the blue chunks.
We sum (\ref{equ:randomchunkmatching}) over 
all  $R + B$ possible chunk matchings and take the average. This gives the 
expected number of matched edges (the sums 
range over all ordered pairs of chunks in $\Gamma$).
\begin{align*}
&\phantom{=} \frac{1}{2(R + B)}
\Bigg(\sum_{C \text{ red}, D \text{ red}} \max \{b(C), b(D)\}
+ 2\sum_{C \text{ red}, D \text{ blue}} k \\
&\qquad + \sum_{C \text{ blue}, D \text{ blue}} \max \{r(C), r(D)\}
\Bigg)\\
\intertext{Since there are $R$ red chunks and
$B$ blue chunks, this is}
&=
\frac{1}{2(R + B)}
\left(\sum_{C \text{ red}, D \text{ red}} \max \{b(C), b(D)\}
+ 
2kRB 
+\sum_{C \text{ blue}, D \text{ blue}} \max \{r(C), r(D)\}\right)\\
\intertext{We lower bound the maximum by
the average to estimate this as}
&\geq
\frac{1}{2(R + B)}
\left(
\sum_{C \text{ red}, D \text{ red}} \frac{b(C) + b(D)}{2}
+ 2kRB +
\sum_{C \text{ blue}, D \text{ blue}} \frac{r(C) + r(D)}{2}
\right)
\tag{**}\\
\intertext{Simplifying the 
sums, this is}
&=
\frac{1}{2(R + B)}
\left(
R \sum_{C \text{ red}} b(C)
+ 2kRB +
B \sum_{C \text{ blue}} 
 r(C)\right)\\
\intertext{Since the total number of blue points
in red chunks is $\alpha kR $ and the 
total number of red points in blue chunks 
is $\beta kB$, this equals}
&=
\frac{\alpha k R^2 + 2kRB + \beta k B^2}{2(R + B)}\\
\intertext{Regrouping the terms and 
using (\ref{equ:chunkpoints}), this
becomes}
&=
\frac{R(\alpha  kR + kB) + B(\beta k B + kR)}{2(R + B)}
 =
\frac{(R + B)n}{2(R + B)}
= \frac{n}{2}.\qedhere
\end{align*}
\end{proof}

\subsection{Taking advantage of \texorpdfstring{$k$}{k}-configurations}

One inefficiency in the calculation in Lemma~\ref{lem:randomchunkmatch}
is that we bound the maximum by the average in inequality
(**). If these two quantities
often differ significantly, we can 
gain an advantage over 
Lemma~\ref{lem:randomchunkmatch}. This is made precise in
the next lemma.
\begin{lemma}
\label{lem:largeindexvariance}
Set $c_4 = 1/40$.
Let $\delta > 0$ and 
let $P$ be a convex bichromatic point set with
$2n$ points, $n$ red and $n$ blue,
and $\Gamma$ a $k$-configuration for $P$ 
with 
index at most $0.11$ that contains
at least $\delta (n/k)$ red chunks or at least
$\delta (n/k)$ blue chunks
with index at least $0.22$. 
Then, $P$ admits a separated
matching of size at least $(1/2 + c_4 \delta^2)n$.
\end{lemma}

\begin{proof}
Suppose without loss of generality that there
are at least $\delta (n/k)$
red chunks with index at least $0.2$.
Let $R$ be the number of red chunks and 
$B$  the number of blue chunks.
The
average red index of $\Gamma$ is at most $0.11$. Thus, if writing
$\gamma_1 (n/k)$ for the number of 
red chunks with index in $(0.11, 0.22)$ and $\gamma_2 (n/k) \geq \delta(n/k)$ 
for the number of red chunks with index in $[0.22, 1)$,
we have
\[
0.11 R \geq 0.11 \gamma_1 \frac{n}{k} + 0.22 \gamma_2 \frac{n}{k}  =
0.11(\gamma_1 + 2\gamma_2)\frac{n}{k}.
\]
It follows that $R \geq (\gamma_1 + 2\gamma_2) (n/k)$,
and there
must be at least 
$\gamma_2(n/k) \geq \delta(n/k)$ red chunks of index in $[0,0.11]$. 
Now, consider the following sum over all ordered pairs $(C, D)$ of red chunks,
where one chunk ($C$ or $D$) has red index at most $0.11$ and the other
chunk ($D$ or $C$) has red index at least $0.22$:
\begin{align*}
&\frac{1}{2(R + B)}
\left(
\sum_{C} \sum_{D} \max \{b(C), b(D)\} -
\frac{b(C) + b(D)}{2}\right)\\
\intertext{Since $2\max\{a,b\} - a - b = 
\max\{a,b\} - \min\{a, b\}$, for all
$a, b \in \R$, this equals}
&=
\frac{1}{4(R + B)}
\sum_{C} \sum_{D} 
(\max \{b(C), b(D)\} - \min\{b(C), b(D)\})
\\
\intertext{One chunk in each summand 
contains at least $0.22k$
blue points, the other chunk contains
at most $0.11k$ blue points, so we 
can lower bound this as}
&\geq
\frac{1}{4(R + B)}
\sum_{C} \sum_{D} 
(0.22 - 0.11)k\\
&\geq
\frac{1}{4(R + B)}
\sum_{C} \sum_{D} 
\frac{k}{10}
\geq
\frac{\delta^2 (n/k)^2}{R + B}  \frac {k} {20}
\geq \frac{\delta^2}{40} n,
\end{align*}
since we are adding 
over at least $2 \delta^2(n/k)^2$ 
ordered pairs $(C, D)$
(recall that each ordered
pair $(C, D)$ has a partner
$(D, C)$ in the
sum) and since by (\ref{equ:chunkupper}), we have
$R + B \leq 2n/k$.
Thus, comparing with (**), 
the lemma follows.
\end{proof}

Lemma~\ref{lem:largeindexvariance} shows that
we can assume that few chunks in the $k$-configuration
$\Gamma$ of $P$ have index larger than $0.22$. 
In fact, suppose now that $\Gamma$ 
contains no chunk of index
at least $0.3$ (this will be justified
below).
From now on, we will also assume that $k$ is divisible by $3$.
We subdivide each chunk in our $k$-configuration $\Gamma$
into three \emph{$(k/3)$-subchunks}. Since all $k$-chunks
have index less than $0.3$, the subchunks have
the same color as the original chunk.
Let $C$ be a $k$-chunk. The \emph{middle subchunk}
of $C$, denoted by $C_M$, is the $(k/3)$-subchunk of 
$C$ that lies in the middle of the three subchunks.
Now, we consider the middle subchunks. If the 
middle subchunks
of the max-index color contain many points of the 
min-index color, we can gain an advantage by considering
two \emph{cross-matchings} between chunks of the max-index color.

\begin{lemma}
\label{lem:uniformmiddlechunks}
Set $c_5 = 1/4$. Let $\delta > 0$ and
let $P$ be a convex bichromatic point set with
$2n$ points, $n$ red and $n$ blue. Let $\Gamma$ be a $k$-configuration for $P$ 
such that (i) $k$ is divisible by $3$;
(ii) 
every chunk in $\Gamma$ has index less than 
$0.3$; and (iii) 
the middle subchunks of the max-index color 
contain in total at least $\delta n$ points of 
the min-index color. Then 
$P$ admits a separated 
matching of size at least $(1/2 + c_5 \delta)n$.
\end{lemma}

\begin{proof}
Suppose that the max-index color is red.
We take a random chunk matching $M$
of $\Gamma$, and we derive a
separated matching from $M$ as described
above. However, when considering 
a pair $(C, D)$ of two red 
chunks, we proceed slightly differently.
First, suppose that $C \neq D$, and
let $C_1, C_2, C_3$ be the three subchunks
of $C$, and $D_1, D_2, D_3$ be the three subchunks
of $D$ (in clockwise order). 
We have $r(C_i) = r(D_i) = k/3$, for $i = 1, 2, 3$;
and $b(C_1) + b(C_2) + b(C_3) < k/3$ and 
$b(D_1) + b(D_2) + b(D_3) < k/3$.
\begin{figure}
    \centering
    \includegraphics{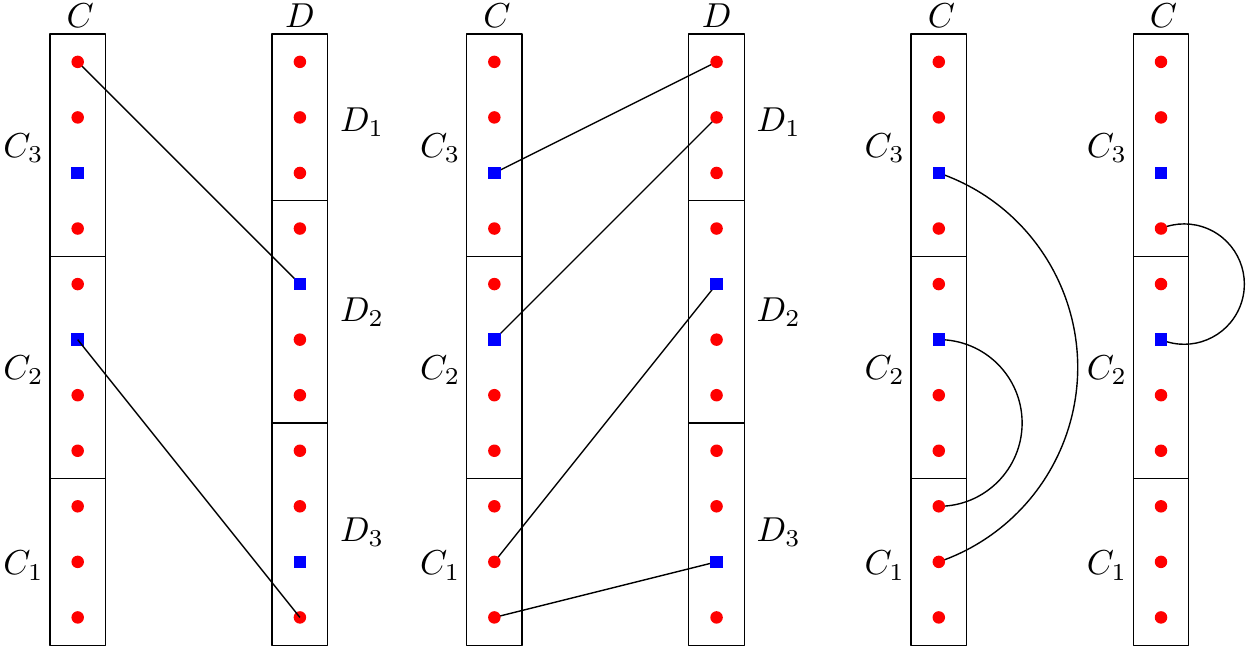}
    \caption{The two different separated matchings
    between two distinct red chunks (left) 
    and the same red chunk (right).}
    \label{fig:hardcase}
\end{figure}
We consider two separated matchings between
$C$ and $D$ (see Figure~\ref{fig:hardcase}(left): (a) match all blue points in 
$C_1$ and $C_2$ to 
red points in $D_3$ and all blue points in $D_1$ and $D_2$
to red points in $C_3$; and (b) 
match all blue points in $D_2$ and $D_3$ to 
red points in $C_1$ and all blue points in $C_2$ and $C_3$
to red points in $D_1$. 
We take the better of the two
matchings. The  number of matched edges $\text{matched}(C, D)$
is lower-bounded by the average,
so
\begin{align}
\text{matched}(C, D) &\geq \frac{1}{2}\left(b(C_1) + b(C_2) + b(D_1) + b(D_2) + 
b(D_3) + b(D_2) + b(C_2) + b(C_3)\right)\notag\\
&= \frac{1}{2}\left(b(C) + b(D)  + b(C_2) + b(D_2) \right).
\label{equ:CDlower}
\end{align}
Second, if $C = D$, we 
subdivide $C$ into 
the three subchunks 
$C_1$, $C_2$, $C_3$ with
$(C_1) = r(C_2) = r(C_3) = k/3$ and $b(C_1) + b(C_2) + b(C_3) < k/3$. Again,
we consider two different
matchings for $C$ (see Figure~\ref{fig:hardcase}(right):
(a)
match the blue points
in $C_1$ and $C_2$ to the
red points in $C_3$, and
(b) match the blue points
in $C_2$ and $C_3$ to the red 
points in $C_1$. Again,
the number of matched edges
$\text{matched}(C, C)$ is 
at least
\begin{equation}
\label{equ:matchedCClower}
    \text{matched}(C, C) \geq \frac{1}{2}
(b(C_1) + b(C_2) + b(C_2) + b(C_3)) 
= \frac{1}{2}(b(C) + b(C_2)).
\end{equation}
Now, we set $R$ to the number of red chunks and $B$
to the number of blue chunks in $\Gamma$. Then, in a random chunk matching, the expected number
of edges in the separated matchings between the pairs
$(C, D)$ of red chunks is
\begin{equation}
\label{equ:matchedAdvantage}
    \frac{1}{2(R + B)} 
\left(\sum_{C \neq D, C, D \text{ red}}
\text{matched}(C, D) +
\sum_{C \text{ red}} 2\,\text{matched}(C, C)\right).
\end{equation}
Note that in the 
first sum, each unordered pair $\{C, D\}$ of
distinct red chunks appears twice, even though
it appears once in a random chunk matching. This
is compensated by the leading factor of $1/2$,
which again leads to a coefficient of
$2$ for the expected number of edges in
the separated matching in a chunk that is
paired with itself. Using
(\ref{equ:CDlower}, \ref{equ:matchedCClower}),
we can write
\[
(\ref{equ:matchedAdvantage})
\geq 
\frac{1}{2(R + B)} 
\left(\sum_{C \text{ red}}\sum_{D \text{ red}}
\frac{
b(C) + b(D) + b(C_M) + b(D_M)}{2}\right),
\]
where we sum over all ordered pairs
$(C, D)$ of red chunks and $C_M$ and $D_M$ denote
the middle chunks of $C$ and $D$.
Now we compare with (**).
\begin{align*}
&\frac{1}{2(R + B)}
\left(
\sum_{C \text{ red}} \sum_{D \text{ red}} \frac{
b(C) + b(D) + b(C_M) + b(D_M)}{2} -
\frac{b(C) + b(D)}{2}\right)\\
&= \frac{1}{2(R + B)}
\left(
\sum_{C \text{ red}} \sum_{D \text{ red}} 
\frac{b(C_M) + b(D_M)}{2}\right)\\
\intertext{In the sum, every middle chunk $C_M$ and
every middle chunk $D_M$ appears exactly $R$ times,
and by assumption, the total number of blue points
in the red middle chunks is at least $\delta n$.
Thus, this is lower-bounded as}
&\geq \frac{1}{2(R + B)}
R \delta n
\geq \frac{1}{4R}
R \delta n
= \frac{\delta}{4}n,
\end{align*}
since red is the max-index color and hence
by Proposition~\ref{prop:confcount},
we have $B \leq R$ and
$R + B \leq 2R$. Thus,
the lemma follows.
\end{proof}

Finally, we consider the case that the middle subchunks
of the max-index color contain relatively few points.
Since the index of $\Gamma$ is relatively small, it 
means that the indices of the middle subchunks of the
max-index color have a large variance. As in
Lemma~\ref{lem:largeindexvariance}, this
leads to a large separated matching.

\begin{lemma}
\label{lem:smallmiddlechunks}
Set $\delta = 10^{-4}$ and $\eps = 10^{-5}$.
Let $P$ be a convex bichromatic point set with
$2n$ points, $n$ red and $n$ blue, and let $\Gamma$ be a $k$-configuration for $P$ 
such that (i) $k$ is divisible by $3$; 
(ii) $\Gamma$  has
index at least $0.09$; and (iii) 
every chunk in $\Gamma$ has index less than 
$0.3$. Then,
if the middle subchunks of the max-index color 
contain in total at most $\delta n$ points of 
the min-index color, 
$P$ admits a separated
matching of size at least $(1/2 + \eps)n$.
\end{lemma}

\begin{proof}
Suppose that the max-index color is red.
Let $R$ be the number of red chunks and $B$
be the number of blue chunks. Denote by
$\alpha \geq 0.09$ the average 
red index.

Since all chunks in $\Gamma$
have index less than $0.3$,
when considering the subchunks, we get
a $(k/3)$-configuration
$\Gamma'$ for $P$. 
We will refer to the
pieces of $\Gamma'$ 
as subchunks, to distinguish
them from the pieces
of $\Gamma$. 
Every red chunk
of $\Gamma$ is partitioned into three
red subchunks of $\Gamma'$, and every blue chunk of
$\Gamma$ is partitioned into three blue 
subchunks of $\Gamma'$.
Thus, $\Gamma'$, 
has $3R$ red subchunks,
$3B$ blue subchunks, and 
the  same
average red index and average blue index
as $\Gamma$.
By Proposition~\ref{prop:confcount},
there are $R \geq n/2k$ middle 
red subchunks.
Thus, there are at least $(n/2k) \cdot (k/3) = n/6$
red points in the middle red subchunks.
By assumption, there are at most $\delta n$ blue 
points in the 
middle red subchunks, so the average index
of the middle red subchunks is at most 
$(\delta n) / (n/6) = 6 \delta$.
By Markov's inequality, 
$\Gamma'$ contains at least 
$0.5R$ red subchunks of index
at most $12 \delta$.

On the other hand, the average red index of $\Gamma'$  is 
 $\alpha \geq 0.09$. Write 
$\gamma \cdot 3R$ for the 
number of red subchunks
with index at least $0.01$. Then,
\[
0.09 \cdot 3R \leq \gamma  \cdot 3R + 0.01 (1 - \gamma)3R
\quad \Rightarrow \quad
0.99\gamma \geq 0.08
\quad \Rightarrow \quad
\gamma \geq 0.08.
\]
Thus, there must be at least $0.08 \cdot 3 R \geq 0.2R$ red subchunks of index at least
$0.01$. 

Consider a random chunk-matching of $\Gamma'$.
and look at the sum over all pairs of red subchunks $(C,D)$ 
where one subchunk ($C$ or $D$) has red index at most 
$12 \delta_4$ and the other
subchunk ($D$ or $C$) has red index at least $0.01$. The advantage
over (**) is
\begin{align*}
&\frac{1}{6(R + B)}
\left(
\sum_{C} \sum_{D} \max \{b(C), b(D)\} -
\frac{b(C) + b(D)}{2}\right)\\
\intertext{Using again that 
$2\max\{a, b\} - a - b = \max\{a, b\} - \min\{a, b\}$,
this is}
&=
\frac{1}{6(R + B)}
\left(
\sum_{C} \sum_{D} 
\frac{\max \{b(C), b(D)\} - \min\{b(C), b(D)\}}{2}
\right)\\
\intertext{Since one chunk in $(C, B)$ contains 
at least $0.01k$ blue points and the other contains
at most $12 \delta k$ blue points, this
is lower bounded as}
&\geq
\frac{1}{12(R + B)}
\sum_{C} \sum_{D} 
(0.01k- 12 \delta k)
\\
\intertext{Since we have at least 
$2\cdot 0.2 R \cdot 0.5 R$ ordered pairs $(C,D)$ of
the desired type, this is}
&\geq \frac{0.2R^2}{12(R + B)} 
\cdot
(0.01-12\delta)k
\geq
\frac{R}{120} (0.01-12\delta)k
\geq \frac{1}{240} (0.01-12\delta) n
\geq \eps n,
\end{align*}
by our choice of $\delta$ and 
$\eps$ and 
since by Proposition~\ref{prop:confcount}, $R  + B \leq 2R$
and $R \geq n/2k$.
\end{proof}

\subsection{Putting it together}

From Theorem~\ref{thm:manyruns}, it follows
that if $P$ has at least four runs, 
there is always a separated matching with 
strictly more than $n/2$ edges. Moreover, 
if $P$ has two runs, then $P$ has 
a separated matching with $n > n/2$ edges.
Therefore, the following theorem implies
Theorem~\ref{thm:main}.

\begin{theorem}
\label{thm:puttingittogether}
There exist constants $\eps_* > 0$ and $n_0 \in \N$
with the following property:
let $P$ be a convex bichromatic point set with $2n \geq 2n_0$ points, $n$ red and $n$ blue.
Then, $P$ admits a separated matching on at least $(1 + \eps_*)n$ vertices.
\end{theorem}

\begin{proof}
Set $n_0 = 10^{100}$
and $\eps  = 10^{-5}$,
as in 
Lemma~\ref{lem:smallmiddlechunks}.
Let $k_1$ the smallest integer larger than 
$10^3 \eps^{-3} = 10^{18}$ that
is divisible by $3$. 
Since $n \geq 10^{100} \geq 8k_1^2$,
Lemma~\ref{lem:smallchunks} 
shows that
if the $(k_1, 0)$-partition $\Gamma_1$  of $P$ has index at least
$0.1$, the theorem follows
with $\eps_* = \Omega(1/k_1^4)
= \Omega(1)$.
Thus, we may assume the 
following claim:
\begin{claim}
The $(k_1, 0)$-partition $\Gamma_1$
of $P$ has index less than $0.1$,
where $k_1$ is a fixed constant
with $k_1 \geq 10^3\eps^{-3}
= 10^{18}$.
\end{claim}
Next,
let $k_2$ be the
largest integer in the
interval $[10^{-4} \eps^3 n, 10^{-3} \eps^3 n]$ that is divisible by $3$. 
Since $n \geq 10^{100}$,
it follows that $k_2$ exists.
Furthermore, since 
$n \geq k_2$ and 
$6480n \leq 10^{-8}\eps^6 n^2 \leq k_2^2$,
Lemma~\ref{lem:largechunks}
implies that
if the $(k_2, 0)$-partition $\Gamma_2$ 
of $P$ has index at most $0.1$, the theorem follows with
$\eps_* = \Omega((k_2/n)^2) = \Omega(1)$.
Hence, we may assume the 
following claim:
\begin{claim}
The $(k_2, 0)$-partition $\Gamma_2$
of $P$ has index more than $0.1$,
where $k_2$ is  the
largest integer in the
interval $[10^{-4} \eps^3 n, 10^{-3} \eps^3 n]$ that is divisible by $3$.
\end{claim}

We now interpolate between $\Gamma_1$ and $\Gamma_2$.
Consider the sequence of $(k, \lambda)$-partitions of $P$ for the
parameter pairs 
\begin{multline*}
(k_1, 0), (k_1, 1),
\dots, (k_1, \lambda(k_1)), (k_1 + 3, 0), (k_1 + 3, 1), 
\dots, \\(k_1 + 3, \lambda(k_1 + 3)), (k_1 + 6, 0), \dots, (k_2, 0),
\end{multline*}
where $\lambda(k)$ denotes the largest $\lambda$ for which
the $(k, \lambda)$-partition of $P$ still contains a $k$-chunk.
Let $(k_*, \lambda_*)$ be the first parameter pair 
for which the index of the $(k_*, \lambda_*)$-partition $\Gamma_3$ 
of $P$ is
larger than $0.1$. This parameter pair exists,
because $(k_2, 0)$ is a candidate.
\begin{claim}
The $(k_*, \lambda_*)$-partition $\Gamma_3$
of $P$ has index in $[0.1, 0.101]$.
Here, $k_*$ is divisible 
by $3$ and lies in the interval
$[10^3\eps^{-3}, 
10^{-3}\eps^3n ]$.
\end{claim}

\begin{proof}
The claim on $k_*$ and the
fact that $\Gamma_3$ has
index at least $0.1$ follow
by construction.
Furthermore, let $(k_{**}, \lambda_{**})$ be such that
$\Gamma_3$ is the 
$(k_{**}, \lambda_{**} + 1)$
partition of $P$ (we either
have $k_{**} = k_*$ and $\lambda_{**} = \lambda_* - 1$; or $k_{**} = k_* - 1$ and $\lambda_{**} = \lambda(k_{**})$). Since
$210000 k_{**} \leq 10^6 \cdot 10^{-3} \eps^{3} n \leq n$, 
Lemma~\ref{lem:smallchange} implies
that the index of $\Gamma_3$ is at most
$0.101$.
\end{proof}

We rearrange $P$ to turn $\Gamma_3$ into
a $k_*$-configuration $\Gamma_4$ of 
a closely related point set 
$P_2$. 

\begin{claim}
\label{clm:Gamma4}
There exists a convex bichromatic
point set $P_2$ with $2n$ 
points, $n$ red and $n$ blue, and a $k_*$-configuration $\Gamma_4$
of $P_2$ such that 
(i) $P_2$ differs from $P$ in
at most $10^{-1} \eps^3 n$ points; and
(ii) the index of $\Gamma_4$ 
lies in $[0.097, 0.103]$.
\end{claim}
\begin{proof}
We remove from $P$ all the 
uncovered points of $\Gamma_3$ as well as $3$ points of the majority color
from each $(k_* + 3)$-chunk of $\Gamma_3$ (and, if necessary, up to $3$ points of the minority color, to keep chunk structure valid).
If we consider a single red
$(k_* + 3)$-chunk $C$ and denote
the original number of blue points
in $C$ by $b(C)$ and the resulting number of blue 
points by $b'(C)$, then
the index of $C$ changes by at most 
\[ 
\left |\frac{b(C)}{k_* + 3} -
\frac{b'(C)}{k_*}\right|
= 
\left|\frac{k_*b(C) - (k_* + 3)b'(C)}
{k_*(k_* + 3)}\right|
\leq
\frac{|b(C) - b'(C)|}{k_* + 3} +
\frac{3b'(C)}{k_*(k_* + 3)}
\leq \frac{6}{k_* + 3},
\]
since $|b(C) - b(C')| \leq 3$ 
and $b'(C) \leq k_*$. A similar bound holds for a blue 
$(k_* + 3)$-chunk.

By (\ref{equ:chunkupper}), 
there are at most $2n/k_* \leq 2 \cdot 10^{-3} \eps^3 n$ many $(k_* + 3)$-chunks, 
and by Proposition~\ref{prop:chunkcount}, 
there at most $2k_* - 1 \leq 2 \cdot 10^{-3}\cdot \eps^3 n$ uncovered points, 
so in total we remove at most $14 \cdot 10^{-3} \eps^3 n \leq 10^{-1} \eps^3 n$ points.
We arrange these points into as many pure chunks of
$k_*$ red points or of $k_*$ blue points as possible.
This creates at most $10^{-1}  \eps^3 (n / k_*)$
new $k^*$-chunks, all of 
which have index $0$.
Now, less than $k_*$ red points and less than $k_*$ blue
points remain. 
By (\ref{equ:chunklower}), there are at 
least 
\[
(1 - 0.101) \frac{n}{7k_*} - 2 
\geq 10^{-1} \cdot 10^3 \eps^{-3} - 2 
\geq 10^3
\]
chunks of
each color in $\Gamma_3$.
Thus,
we can partition the remaining red points into at most $10^3$ groups 
of size at most $10^{-3} k_*$ and add each group to a
single blue chunk; 
and similarly for the remaining blue points. This changes the
index of each chunk by at most 
$10^{-3}$.

We call the resulting rearranged point set $P_2$ and the resulting
$k_*$-configuration $\Gamma_4$. As mentioned,
$P_2$ was obtained from $P$ by moving at most $10^{-1}
\cdot \eps^3 n$ points.
We change the 
index of any existing chunk by at most 
$6/(k^* + 3) + 10^{-3} \leq 2 \cdot 10^{-3}$. 
Furthermore, we create at 
most $10^{-1}  \eps^3 (n/k_*)$ new $k_*$-chunks (all of index $0$) and by (\ref{equ:chunklower}),
we have at least 
$(1-0.101)n/(7k_*) - 2 \geq 
(10^{-1} - 10^{-2} \cdot \eps^3)(n/k_*)$ original chunks of each color in $\Gamma_3$.
Thus, if we denote by $\alpha$
the average index of the
existing red chunks after
the rearrangement,
by $R$ the number of existing 
red chunks, and by $R'$ the 
number of new red chunks,
the average red index of $\Gamma_4$
can differ from $\alpha$ by at most
\[
\alpha - \frac{R}{R + R'} \alpha
= \alpha \frac{R'}{R + R'}
\leq \alpha \frac{R'}{R}
\leq 0.102 \frac{10^{-1} \eps^3}
{10^{-1} - 10^{-2} \eps^3}
\leq 10^{-3},
\]
and similarly for the average
blue index of $\Gamma_5$.
It follows that $\Gamma_4$ has
index in $[0.097, 0.103]$.
\end{proof}

Now, using Lemma~\ref{lem:largeindexvariance} with
$\delta = 10^{-1}\eps$, we get that if the $k^*$-configuration $\Gamma_4$ contains 
at least $\delta(n/k_*)$ red chunks or at least $\delta(n/k_*)$
blue chunks with index at least $0.22$, then the rearranged point
set $P_2$ admits a separated matching of size 
at least 
\[
\left(\frac{1}{2} + \frac{1}{40}
\cdot 10^{-2} \eps^2\right)n 
\geq \left(\frac{1}{2} + 10^{-4} \cdot \eps^2\right)n.
\]
By Claim~\ref{clm:Gamma4},
$P_2$ differs from $P$
by at 
most $10^{-1} \eps^3 n$ points. 
Since
$\eps  = 10^{-5}$, 
it follows that after deleting 
all matching edges incident to  a rearranged point,
we obtain the theorem. 
Thus, we may assume the
following claim:
\begin{claim}
At most $10^{-1}\cdot \eps(n/k_*)$ red chunks and at most
$10^{-1} \cdot \eps(n/k_*)$ blue chunks in $\Gamma_4$ have index
more than $0.22$.
\end{claim}

We again rearrange the point set $P_2$ to
obtain a point set $P_3$ and a $k^*$-configuration $\Gamma_5$
for $P_3$ such that every $k^*$-chunk in $\Gamma_5$ has
index less than $0.3$.
\begin{claim}
\label{clm:Gamma5}
There exists a convex bichromatic
point set $P_3$ with $2n$ 
points, $n$ red and $n$ blue, and a $k_*$-configuration $\Gamma_5$
of $P_3$ such that 
(i) $P_3$ differs from $P_2$ in
at most $2 \cdot 10^{-1} \eps n$ points; 
(ii) the index of $\Gamma_5$ 
is at least $0.096$; 
(iii) all chunks in $\Gamma_5$
have index less than $0.3$; and
(iv) $k_*$ is divisible by $3$.
\end{claim}
\begin{proof}
We remove all the blue points from red chunks of index at 
least $0.22$ and all the red points from all blue chunks
of index at least $0.22$. These are at most 
$2 \cdot 10^{-1} \cdot \eps n$ 
points in total. 
By removing these points,
we decrease the index of
at most $10^{-1} \eps (n/k_*)$  existing chunks of each color to $0$.
By Proposition~\ref{prop:confcount},
there are at least
\begin{equation}
\label{equ:exchunks}
(1 - 0.103) \frac{n}{2k_*} \geq 10^{-1} \cdot \frac{n}{k_*}
\end{equation}
existing chunks of each
color, so this step
decreases the average index
by at most $\eps$.

We rearrange the deleted points into as 
many pure chunks with $k_*$ red points or with 
$k_*$ blue points as possible. Less than $k_*$ red points 
and  less than $k_*$ blue points remain. By~(\ref{equ:exchunks}),
there are at least 
$10^{-1}(n/k_*) \geq 10^3$
chunks of each color,
so  we group the remaining
points into blocks of size $10^{-3}\cdot k_*$ 
and distribute the blocks over the
existing red and blue chunks.
This increases the average index
of the existing chunks
by at most $10^{-3}$.

Finally, we create at most 
$10^{-1} \cdot \eps (n/k_*)$ new chunks of each color (all with index $0$), and the existing number of 
chunks of the max-index color of $\Gamma_4$ is at least $n/2k_*$,
by Proposition~\ref{prop:confcount}.
Suppose for concreteness that
the max-index color of $\Gamma_4$ 
is red, and let $R$ be the
number of existing red chunks,
$R'$  the number of new red chunks,
and 
$\alpha$  the average
index of the existing red chunks
after the rearrangement.
Then, the average red index 
after the rearrangement differs
from $\alpha$ be at most
\[
\alpha - \frac{R}{R + R'} \alpha 
\leq \alpha \frac{R'}{R}
\leq 0.104 \cdot  \frac{  10^{-1} \eps}{1/2} \leq \eps.
\] 
Thus, the red index in 
the resulting $k_*$-configuration 
$\Gamma_5$ is
at least $0.097 - 2\eps \geq 0.096$. This implies that the index of $\Gamma_5$ is
at least $0.096$.
\end{proof}

Now, we consider 
the $k^*$-configuration
$\Gamma_5$. By Lemma~\ref{lem:smallmiddlechunks},
if in $\Gamma_5$ the middle-chunks of the max-index color
contain in total at most $10^{-4} n$ points of the min-index color,
we get a separated matching for $P_3$ of size
at least $(1/2 + \eps)n$. By deleting all the matching edges
that are incident to the at most $2 \cdot 10^{-1} \eps n + 10^{-1} \eps^3 n
\leq 0.3 \eps n$ points that were moved to obtain $P_3$
from $P$, the theorem follows. Similarly, if in $\Gamma_5$ the
middle-chunks of the max-index color contain in total 
 more than $10^{-4} n$ points of the min-index color,
by Lemma~\ref{lem:uniformmiddlechunks},
we get a separated matching for $P_3$ of size at least
$(1/2 + 10^4/4)n \geq (1/2 + \eps)n$.
Again, we obtain the theorem after deleting edges that are incident to 
the rearranged points.
\end{proof}

\section{Existence of large separated monochromatic matchings}
We outline the proof of Theorem~\ref{thm:main-monochromatic}.
This goes in two steps. First, we consider the case
that $P$ has the same number of red and blue 
points, and we derive a counterpart to 
Theorem~\ref{thm:puttingittogether} for it.
The main ideas are the same as
for the proof of Theorem~\ref{thm:main}.
Then, we show how this can be extended to the
case that the number of red and blue points differs.

\subsection{The balanced case}

First, we suppose that the number of
red points and the number of  blue points in $P$
is exactly $n$.
We again consider $k$-chunks as in
Section~\ref{sec:chunks}, and
we use
random chunk-matchings as
explained in Section~\ref{sec:chunkmatching}. Suppose that $k$ is divisible
by $2$.
We derive  a separated 
monochromatic matching 
from a chunk matching $M$
as follows.
Suppose two chunks $C$ and $D$ are matched in $M$. If $C=D$, we find $k/2$ pairwise disjoint edges with endpoints in 
the same (major) color. Now suppose that $C \neq D$. If $C$ and 
$D$ are both blue  or both red, we  
take $k$ pairwise disjoint edges between them, using 
points of their major color. If, say, $C$ is red and 
$D$ is blue,  we  may either take $b(C)$ blue edges or 
$r(D)$ red edges that are pairwise disjoint and 
connect points of $C$ with points of $D$. Thus, we obtain 
$\max\{b(C),r(D)\}$ edges between $C$ and $D$.
Similarly to 
(\ref{equ:randomchunkmatching}), this gives a separated monochromatic
matching with
\begin{multline}
\label{equ:randomchunkmatching-monochr}
\frac{1}{2}
\Bigg(\sum_{\substack{(C, D) \in M \\ C\text{ red}, D \text{ red}}} k
+
\sum_{\substack{(C, D) \in M \\ C\text{ red}, D \text{ blue}}} 
\max\{b(C), r(D)\}\\
+
\sum_{\substack{(C, D) \in M \\ C\text{ blue}, D \text{ red}}} \max \{r(C), b(D)\}
+
\sum_{\substack{(C, D) \in M \\ C\text{ blue}, D \text{ blue}}}  k
\Bigg)
\end{multline}
edges, where the sums go over all ordered pairs of matched chunks in $M$,
i.e., a matched pair $(C, D)$ with $C \neq D$ appears twice (which is compensated
by the leading factor of $1/2$) and a matched 
pair $(C, C)$ appears once.
The following lemma is analgous
to Lemma~\ref{lem:randomchunkmatch}.

\begin{lemma}
\label{lem:randomchunkmatch-monochromatic}
Let $k$ be even, and
let $\Gamma$ be a $k$-configuration in $P$. 
Let $M$ be  a random chunk matching $M$ in $\Gamma$.
The expected number of 
edges in 
the corresponding separated monochromatic matching is at least 
$n/2$.
\end{lemma}

\begin{proof}
Let $R$ be the number of red chunks in $\Gamma$ and
let $B$ be the number of blue chunks in $\Gamma$.
Let $\alpha$ be the average index of the red chunks,
and let $\beta$ be the average index of the blue
chunks.
We sum (\ref{equ:randomchunkmatching-monochr}) over 
all  $R + B$ possible chunk matchings and take the average. We get
that the expected number of matched edges is at least (the sums 
range over all ordered pairs of chunks in $\Gamma$)
\begin{align*}
&\phantom{=} \frac{1}{2(R + B)}
\Bigg(\sum_{C \text{ red}, D \text{ red}} k
+
\sum_{C \text{ red}, D \text{ blue}} \max \{b(C), r(D)\}\\
&\ \ +
\sum_{C \text{ blue}, D \text{ red}} \max \{r(C), b(D)\}
+ \sum_{C \text{ blue}, D \text{ blue}}  k  \Bigg). \\
\intertext{There are
$R^2$ pairs of red chunks and
$B^2$ pairs of blue chunks,
and the maximum can be lower bounded
by the average, so this is}
&\geq \frac{1}{2(R + B)}
\Bigg( 
kR^2 
 +
\sum_{C \text{ red}, D \text{ blue}} \frac{b(C) + r(D)}{2}
+ \sum_{C \text{ blue}, D \text{ red}} \frac{r(C) + b(D)}{2} 
+ kB^2\Bigg)
\tag{***}\\
\intertext{Simplifying the sums,
this becomes}
&=
\frac{1}{2(R + B)} 
\Bigg(kR^2 + 
B \sum_{C \text{ red}}
  b(D) 
+ R \sum_{C \text{ blue}} r(C) 
+kB^2\Bigg) \\
\intertext{Since there are
$\alpha k R$ blue points
in the red chunks, and
$\beta b B$ red points in the 
blue chunks, this is}
&=
\frac{kR^2 + 
\alpha k BR   
+ \beta k BR  + k B^2}{2(R + B)} \\
\intertext{Regrouping the
terms and using (\ref{equ:chunkpoints}), this
equals}
&=
\frac{R(kR + \beta k B) +
B(kB  + \alpha kR)}{2(R + B)} 
=
\frac{(R + B)n}{2(R + B)}
= \frac{n}{2}.
\end{align*}
\end{proof}

The other lemmas and 
theorems from
Section~\ref{sec:bichromatic}
have their counterparts for
monochromatic matchings which
can be always obtained by
changing the words ``separated
matching'' to the words
``separated monochromatic
matching'' in the statement. We
briefly describe the proof idea
for each of these new lemmas.

\begin{itemize}
\item
In the proof of the counterpart
of Theorem~\ref{thm:manyruns},
$M_i'$ is the separated
monochromatic submatching of
$M_i$ consisting of the
monochromatic edges of $M_i$.
Again, the average size of
$M_i'$ can be increased
to $n/2 + \Omega(t^2/n)$ by
adding appropriate
(monochromatic) edges.

\item
For the proof of the counterpart
of Lemma~\ref{lem:smallchunks},
we proceed similarly as in the
proof of 
Lemma~\ref{lem:smallchunks}.
Instead of
Theorem~\ref{thm:manyruns}, we
apply its counterpart.

\item Assumptions in the
counterpart of
Lemma~\ref{lem:largechunks}
imply that there is a chunk $D$
where, say, the number of red
points
exceeds the number of blue
points by a linear additive term. It is
then easy to find the required
large separated monochromatic
matching by matching (almost)
all red points in $D$ and
(almost) all those points in the
complement of $D$ which have the
color which is more frequent in
the complement of $D$.

\item The counterpart of
Lemma~\ref{lem:smallchange} does
not differ from
Lemma~\ref{lem:smallchange},
thus it is already proved.

\item In the proof of the
counterpart of Lemma~\ref{lem:largeindexvariance} 
we proceed similarly as in
Lemma~\ref{lem:largeindexvariance}.

\item In the proof of the
counterpart of
Lemma~\ref{lem:uniformmiddlechunks}, we take a random chunk
matching
of $\Gamma$. However, when
matching a red 
chunk $C$ and a blue chunk $D$,
we consider the following two
separated matchings between
$C$ and $D$: (a) match all blue
points in $C_1$ and $C_2$ to 
blue points in $D_1$ and all red
points in $D_2$ and $D_3$
to red points in $C_3$; and (b) 
match all red points in $D_1$
and $D_2$ to 
red points in $C_1$ and all blue
points in $C_2$ and $C_3$
to blue points in $D_3$.
In the rest of the proof,
we proceed similarly as in the
proof of Lemma~\ref{lem:uniformmiddlechunks}.

\item In the proof of the
counterpart of
Lemma~\ref{lem:smallmiddlechunks}, 
we partition each $k$-chunk into three $(k/3)$-chunks.
Similarly as in the proof of Lemma~\ref{lem:smallmiddlechunks},
a disbalance of the indices of the $(k/3)$-chunks leads
to the desired lower bound.
\end{itemize}

Since all the lemmas in
Section~\ref{sec:bichromatic}
have 
counterparts for separated
monochromatic matchings, we can
derive the following monochromatic
counterpart of
Theorem~\ref{thm:puttingittogether}.

\begin{theorem}
\label{thm:monochromaticbalanced}
The are constants $\eps > 0$ and 
$n_0 \in \N$  such that
any set $P$ of $2n \geq n_0$ points
in convex position, $n$ red and $n$ blue,
admits
a separated monochromatic matching with at
least $n/2 + \eps n$ edges.
\end{theorem}

\subsection{The general case}

We derive Theorem~\ref{thm:main-monochromatic}
from Theorem~\ref{thm:monochromaticbalanced}.
Suppose that $P$ contains $r$ red points and
$b$ blue points, i.e.,  $n = r + b$.
If $r = b$, we are done by
Theorem~\ref{thm:monochromaticbalanced}.
Thus, assume (without loss of generality) that
$r > b$. We distinguish two cases. For this,
let $\eps \leq 1/2$ and $n_0$ be the constants from Theorem~\ref{thm:monochromaticbalanced}.
We assume that $n \geq \max\{n_0/(1-\eps^2), 4/\eps^2\}$.

First, suppose that that that $r \leq b + \eps^2 n$.
We delete $r - b \leq \eps^2 n$ points from
$P$ to obtain a balanced set $P'$ with
\[
n' \geq (1 - \eps^2)n \geq n_0
\]
points. By Theorem~\ref{thm:monochromaticbalanced},
we get a monochromatic separated matching  $M$
on at least
\[
\frac{n'}{2} + \eps n'
\geq \frac{(1 - \eps^2)n}{2} + 
\eps(1 - \eps^2) n
=  \frac{n}{2} + (\eps - \eps^2/2 - \eps^3) n
\geq \frac{n}{2} + \frac{\eps n}{2},
\]
vertices.\footnote{Note the subtlety
that the here we express the size of the matching
in the number
of \emph{vertices}, while Theorem~\ref{thm:monochromaticbalanced} talks about
the number of \emph{edges}. This is compensated
by the fact that Theorem~\ref{thm:monochromaticbalanced} is
applied with $n = n'/2$.}
Clearly, $M$ is also a monochromatic
separated matching for $P$.

Second, suppose that $r > b + \eps^2 n$.
By greedily pairing the red points, we obtain
a monochromatic separated matching on
\[
2\left\lfloor\frac{r}{2} \right\rfloor
\geq r - 1 = \frac{r + b}{2} +
\frac{r - b}{2} - 1
\geq \frac{n}{2} +
 \frac{\eps^2n}{2} - 1
 = \frac{n}{2} - \frac{\eps^2n - 2}{2}
 \geq \frac{n}{2} + \frac{\eps^2 n}{4},
 \]
vertices, since $n \geq 4/\eps^2$ and hence $2 \leq \eps^2n/2$.

\section{Conclusion}

We have obtained the first improvement over the
simple lower bound bound on the size of a separated
monochromatic or bichromatic matching of an additive
term that is $\Omega(n)$. However, our result is only meaningful
in a qualitative sense, giving constants that very small.
We have made no effort to optimize the constants in our
proof, favoring simplicity. It may be worthwhile
to find out how far our approach can be pushed.

\bibliography{references}

\end{document}